\documentclass[10pt,journal]{IEEEtran}

\usepackage{cite}

\usepackage{graphicx}
\usepackage{subfigure}
\usepackage{epstopdf}
\usepackage{multicol}
\usepackage{algorithm}
\usepackage{algorithmic}
\usepackage{amsthm,amsmath,amsfonts}
\usepackage{mathrsfs}
\usepackage{courier}
\newtheorem{lemma}{\textbf{Lemma}}

\newtheorem{prop}{\textbf{Proposition}}
\usepackage{bbm}
\usepackage{bm}
\usepackage{subfigure}

\usepackage{color}

\begin{document}
%
\title{Fresh, Fair and Energy-Efficient Content Provision in a Private and Cache-Enabled UAV Network}
%
%
%
\author{Peng Yang,~\IEEEmembership{Member,~IEEE}, Kun Guo,~\IEEEmembership{Member,~IEEE}, Xing Xi,~\IEEEmembership{Graduate Student Member,~IEEE},
Tony Q. S. Quek,~\IEEEmembership{Fellow,~IEEE},
Xianbin Cao,~\IEEEmembership{Senior Member,~IEEE}, and Chenxi Liu,~\IEEEmembership{Member,~IEEE}
\thanks{
P. Yang, K. Guo, and T. Q. S. Quek are with the Information Systems Technology and Design, Singapore University of Technology and Design, 487372 Singapore.
X. Xi, and X. Cao are with the School of Electronic and Information Engineering, Beihang University, Beijing 100083, China.
C. Liu is with the State Key Laboratory of Networking and Switching Technology, Beijing University of Posts and Telecommunications, Beijing 100876, China.
This paper was presented in part in the IEEE Global Communications Conference 2020 \cite{yang2020repeatedly}.
}
}

\maketitle

\begin{abstract}
In this paper, we investigate a private and cache-enabled unmanned aerial vehicle (UAV) network for content provision. Aiming at delivering fresh, fair, and energy-efficient content files to terrestrial users, we formulate a joint UAV caching, UAV trajectory, and UAV transmit power optimization problem. This problem is confirmed to be a sequential decision problem with mixed-integer non-convex constraints, which is intractable directly. To this end, we propose a novel algorithm based on the techniques of subproblem decomposition and convex approximation. Particularly, we first propose to decompose the sequential decision problem into multiple repeated optimization subproblems via a Lyapunov technique. Next, an iterative optimization scheme incorporating a successive convex approximation (SCA) technique is explored to tackle the challenging mixed-integer non-convex subproblems. Besides, we analyze the convergence of the proposed algorithm and derive the theoretical value of the expected peak age of information (PAoI) to estimate the content freshness. Simulation results demonstrate that the proposed algorithm can achieve the expected PAoI close to the theoretical value and is more $22.11$\% and $70.51$\% energy-efficient and fairer than benchmark algorithms.
\end{abstract}


%
\IEEEpeerreviewmaketitle

\section{Introduction}
%
%
%
%
\IEEEPARstart{T}{he} data traffic requested by terrestrial mobile users will increase dramatically in terrestrial wireless mobile communication networks \cite{ forecast2019cisco}. It is predicted that monthly data traffic in the global mobile networks will reach 77 exabytes by 2022 \cite{forecast2019cisco}.
It is also foreseeable that with the advancement of manufacturing, chips, and sensors technologies, the global data traffic in future wireless networks will increase exponentially \cite{ji2020joint}. However, the flexibility and resilience of terrestrial network services are insufficient \cite{cao2018airborne}. It is a challenging task for terrestrial networks to guarantee satisfactory network performance at any time, especially during peak traffic time \cite{wu2020optimal}.

Owing to the agile and resilient deployment, the unmanned aerial vehicle (UAV) network has been widely considered as a significant complement in 5G and beyond to terrestrial networks to boost the capacity of terrestrial networks and extend the network coverage \cite{cao2018airborne}.
Moreover, by deploying a private UAV network with complete control over some aspects (e.g., network resources and storage resources) of the network, the private UAV network can provide further optimized services over the service area \cite{aloqaily2021design}.
Recently, the research on the private UAV network attracts much attention from academia and industry \cite{aloqaily2021design,urama2020uav,amorim2019forecasting,de2018optimal}.

On the other hand, UAV caching is a promising paradigm to assist terrestrial networks \cite{jiang2018multimedia}. By proactively caching popular and repetitively requested content files with large size (e.g., high-resolution map, football match video), UAV caching can significantly alleviate the traffic burden and backhaul congestion of terrestrial networks in the peak hours of some hotspots \cite{ji2020joint,wu2020optimal}.
Besides, when content requests are hit by the caching, content files can be directly transmitted without traversing wireless backhaul, which reduces the response delay significantly \cite{wang2017multi}. On-demand UAV communications can also be dispatched when terrestrial networks are overloaded, the manner of which is flexible and cost-effective.
As a result, during the past few years, the issue of UAV caching has been studied extensively \cite{ji2020joint,wu2020optimal,ji2020Design,chai2020online,kalantari2020wireless}.

\subsection{Related work}
In terms of the research on the private UAV network, the work in \cite{aloqaily2021design} deployed a private blockchain-enabled UAV 5G network to meet dynamic user demands in a reliable and secure manner.
In \cite{urama2020uav}, the UAV-aided interference assessment for private 5G new radio (NR) deployments was investigated.
The utilization of dedicated portions of cellular spectrum to provide the high-reliable command and control link for UAVs was evaluated in \cite{amorim2019forecasting}. Besides, the work in \cite{de2018optimal} designed optimal trajectories of UAVs in private UAV networks to always maintain connections between UAVs and a ground station under the constraint that the total distance travelled by all UAVs is minimum.

In terms of the research on the UAV caching, the work in \cite{ji2020joint} explored a cache-enabled UAV assisted wireless network to maximize the minimum throughput among UAV served users, by jointly optimizing the cache placement, UAV trajectory, and UAV transmit power in a finite flight period.
In \cite{wu2020optimal}, the UAV-aided edge caching to assist terrestrial vehicular networks in delivering high-bandwidth content files was investigated. 
Besides, the issue of UAV caching for decreasing the transmission latency and alleviating backhaul congestion in a UAV network with limited wireless backhaul capacity was studied in \cite{kalantari2020wireless}.
Although the UAV caching issue was extensively discussed in \cite{ji2020joint,wu2020optimal,ji2020Design,chai2020online,kalantari2020wireless} via optimizing cache placement and UAV trajectory, and so on, few of them discuss the problem of maintaining the ``freshness" of cached contents. It is crucial to deliver fresh content files to the destination nodes in some applications especially in some delay-sensitive applications, such as intelligent transportation, environmental monitoring, and health monitoring \cite{hu2020aoi}. The outdated information may result in degraded user experience, erroneous control, even cause big catastrophes \cite{samir2020online}.
The ``freshness" is an important metric, referred to as the age of information (AoI) or status age, which is defined as the amount of time elapsed since the instant at which the most recently delivered update takes place \cite{kaul2012real}.
Due to the significance of delivering fresh content files, AoI-aware UAV-assisted wireless network design has attracted increasing interest from the research community \cite{hu2020aoi,samir2020online,zhang2020age}.
For example, the work in \cite{hu2020aoi} proposed to jointly optimize the UAV trajectory, the time required for energy harvesting, and data collection for each sensor node to minimize the average AoI of collected data in a UAV-assisted wireless network. The work in \cite{zhang2020age} investigated the problem of jointly optimizing the UAV trajectory, sensing time, transmission time, and task scheduling to guarantee the freshness of the UAV sensed data in a UAV network. 

\subsection{Motivation and contributions}
Despite the extensive existing work on fresh data provision, all of them \cite{hu2020aoi,samir2020online,zhang2020age} focused on the issue of maintaining the freshness of data delivered to a single destination node from a single/multiple source nodes without discussing the significant problem of minimizing the AoI of data destined to multiple destination nodes.
In fact, the investigation on the issue of delivering fresh data to multiple destination nodes is significant.
On one hand, multiple UAVs instead of a single UAV should be deployed to deliver data to destination nodes owing to the limited service ability and coverage range of a single UAV. On the other hand, for delay-sensitive information (e.g., traffic information and live news), it is important to guarantee the freshness of delivered data.

The goal of minimizing the AoI of data delivered to multiple destination nodes, however, poses novel and greater challenges to the optimization of AoI and the theoretical analysis of AoI in a UAV-assisted wireless network.
First, the minimization of AoI of data depends on the simultaneous decision making on data placement, data delivery, UAV trajectory, and UAV transmit power \cite{hu2020aoi,samir2020online,zhang2020age}. Multiple UAVs rather than a single UAV should be deployed to deliver fresh data to destination nodes. 
Nevertheless, the joint decision making on data placement, data delivery, UAV trajectory, and UAV transmit power in a multi-UAV network is much more challenging than that in a single UAV network studied in \cite{hu2020aoi,samir2020online,zhang2020age}.
Besides, this joint decision making problem can be confirmed as a sequential decision problem with a high-dimensional and mixed discrete and continuous decision space. Therefore, it is difficult to explore some standard optimization methods and the popular deep reinforcement learning methods, which are designed either for the low-dimensional discrete decision space or the high-dimensional continuous decision space, to solve this problem.
Although we considered multiple destination nodes for a joint decision making on UAV trajectory and UAV transmit power in our previous work \cite{yang2020repeatedly}, the issues of ensuring fresh data delivery was not addressed.
Second, owing to the complex interaction between multiple decision variables and the freshness of data, the theoretical analysis on the freshness of data destined to multiple destination nodes has not been well investigated.

To tackle the above challenges and provide further optimized services, we focus on the joint design of UAV caching (including content placement and content delivery), UAV trajectory, and UAV transmit power in a private and cache-enabled UAV network in this paper. The main contributions of this paper can be summarized as follows:

1) A time-varying UAV network is desired to be deployed to deliver fresh content to terrestrial users due to the limited UAV service ability and communication coverage range. In this regard, we formulate the problem of content provision by deploying a private and cache-enabled UAV network as a sequential decision problem. The goal of this problem is to maintain the freshness of data arriving at all users and provide fair and energy-efficient content delivery for all users, subject to UAVs' transmit power and trajectory constraints.

2) To effectively solve the formulated problem, a Lyapunov-based optimization framework and a novel algorithm with provable performance guarantees are proposed. Particularly, the framework solves the complicated sequential decision problem via decomposing it into repeatedly optimized subproblems of multi-tier structure rather than solve it as a whole. The decomposed subproblems, however, are confirmed to be mixed-integer non-convex, which are still intractable directly. To make them tractable, the proposed algorithm firstly explores an iterative optimization scheme to handle the mixed-integer issue. Then, a successive convex approximation (SCA) technique is leveraged to tackle the non-convexity.

3) Besides, we analysis the convergence of the proposed algorithm. The theoretical value of the expected peak AoI (PAoI) to estimate the freshness of the content is also obtained by the probability theory.

4) Finally, the performance of the proposed algorithm is compared with different benchmark algorithms, and impact of different design parameters is discussed.
    Simulation results verify that the proposed algorithm can achieve the expected PAoI close to the theoretical value. Further, the proposed algorithm is more $22.11$\% and $70.51$\% energy-efficient and fairer than benchmark algorithms, respectively.

\section{System model and problem formulation}
\subsection{Scenario description}
This paper considers a content server and a private and cache-enabled UAV network, which includes one ground base station (BS), multiple UAVs, and many terrestrial mobile users, as shown in Fig. \ref{fig_scenario}. The content server will proactively transmit content files, each of which consists of many data packets, requested by users to the BS such that the latency for users to obtain content files can be significantly reduced.
\begin{figure}[!t]
\flushleft
\includegraphics[width=2.0in]{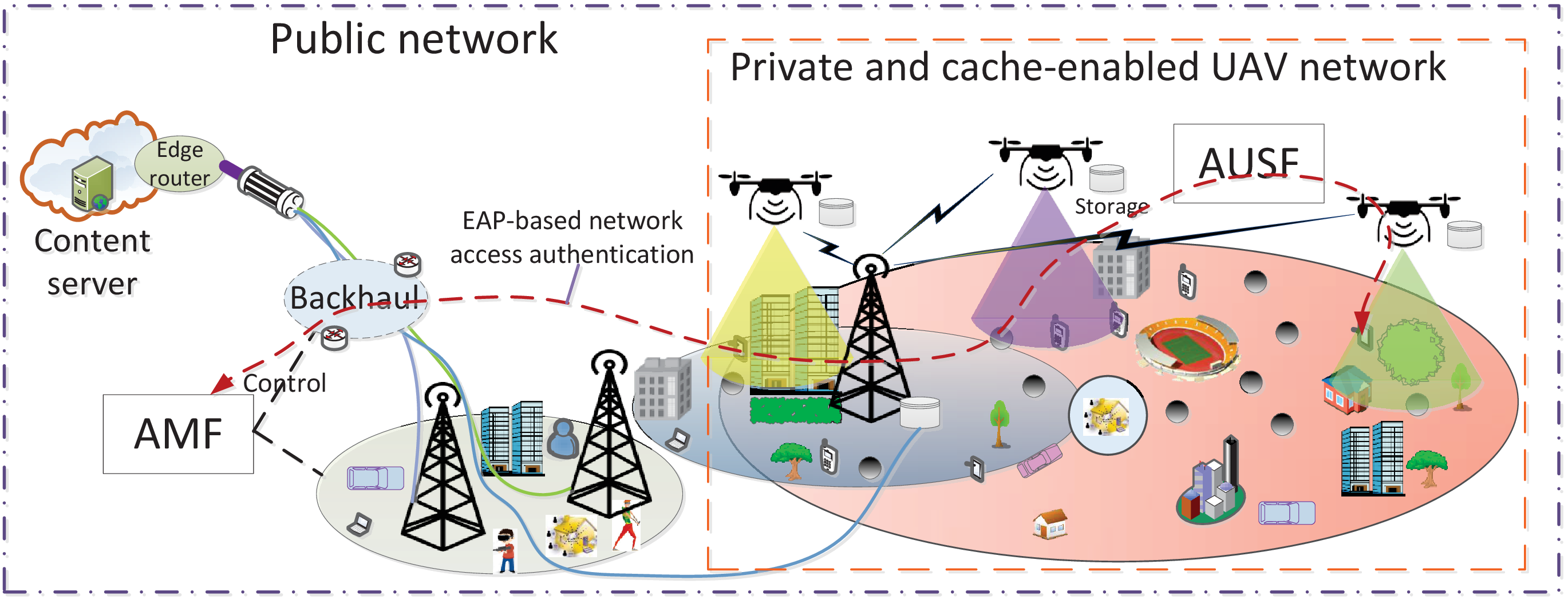}
\caption{The communication scenario and network architecture of a private and cache-enabled UAV network.}
\label{fig_scenario}
\end{figure}
However, terrestrial users, the set of which is ${\mathcal I} = \{1, 2, \ldots, N\}$ with $N$ being the number of users, may be in poor communication environment due to serious signal occlusion or too far from the BS. Then, the BS-user transmission links may be interrupted, and users will have poor quality-of-experience (QoE).
Hence, a set ${\mathcal J}$ of $J$ energy-constrained rotary-wing UAVs acting as aerial relays is deployed to perform the communication task of proving fresh, fair and energy-efficient content delivery for terrestrial users.
In addition, to provide better direct content delivery services for terrestrial users, each UAV is equipped with a capacity-limited storage to dynamically cache content files from the BS and then deliver cached files to users.

To theoretically model the communication task, the time domain in the private and cache-enabled UAV network is assumed to be discretized. We consider a general case of assuming that the UAV communication task of delivering content files may last long enough, i.e., $t = \{1, 2, \ldots\}$, and $\Delta_ t$ is the duration of a time slot.
Owing to the limited number of UAVs and a UAV's restricted communication range, the locations of UAVs should be continuously adjusted when executing the task such that UAVs can deliver content files to all terrestrial users.


\subsection{Network architecture}
In the above scenario, terrestrial users served by the UAVs need to access the content server residing in the public network. To this aim, we design a network architecture shown in Fig. 1. This network architecture includes a public network, a private and cache-enabled UAV network, and a BS shared by the public network and the private and cache-enabled UAV network (referred to as UAV network for brevity).
For the BS, it is split into higher radio access layers (i.e., radio resource control (RRC), packet data control protocol (PDCP)) and lower-layer radio interface (i.e., radio link control (RLC), medium access control (MAC), physical layer (PHY)). The higher layers can be configured to operate in a user-specific mode. For example, for a user requiring low-latency services, RRC can be configured and tailored to disable the lower-layer Internet protocol (IP) stack and related header compression. Besides, RLC can be configured in transparent mode by the RRC. In contrast, for a user requiring high quality of experience (QoE), IP and acknowledged RLC should be initiated \cite{rost2017network}.
The UAV network and the public network will share the lower-layer radio interface of the BS. The higher radio access layers of the BS are split to the public network. Besides, the public network will provide core network functions.

To access the content server, a roaming agreement should be established between the public network and users in the UAV network \cite{schneider2018providing}. This can be done by using a communication terminal or else just by using the subscriber identity in the public network.
Yet, when sharing the BS, the network security is one of the major concerns \cite{schneider2018providing}.
In 3GPP specified 5G networks, an extensible authentication protocol (EAP) based method (see \cite{aboba2004extensible}) can be introduced to ensure the security. In the considered scenario, the authentication is done between users and an authentication server function (AUSF) in the UAV network. This leads to keys shared between the AUSF and users. The keys are derived during authentication and are utilized to protect user traffic.
According to \cite{Security2020GPP}, one of these keys is passed on towards the access and mobility management function (AMF) in the public network as a basis for establishing security between the mobile device and the public network. Another key agreed with the mobile during the same EAP run can however be retained by the private network and can be used to establish user plane security between users and the UAV network.
Besides, network security is terminated by the PDCP. Making the split between the UAV network and the public network below the PDCP allows to keep all security related functions within the UAV network.

To instantiate the network architecture, the network function virtualization (NFV) and software-defined networking (SDN) techniques should be explored.
NFV and SDN are complementary technologies that achieve the level of abstraction and flexibility required to satisfy stringent applications' requirements while maximizing network infrastructure reutilization \cite{etsi2013network}.
Specifically, NFV can decouple physical network functions (PNFs) (e.g., firewalls, routers, load balancers) from dedicated hardware by implementing the same functionality in software, called virtualized network functions (VNFs) \cite{etsi2013network}.
VNFs may then be instantiated in data centers at backend clouds, or on top of devices equipped with compute and storage resources at the edge \cite{taleb2017multi}. For example, in the considered architecture, candidates for VNF instances include AUSF and AMF.

SDN can decouple the user plane from the control plane and centralize network management in an SDN controller. The network management can then be facilitated via a softwarization approach.
With a global view of the network resources, SDN controller applications can take advantage of numerous southbound interfaces (e.g., OpenFlow, NETCONF) to gather network state information and act upon each forwarding device (PNF or VNF) accordingly.
In the considered architecture, due to the exploration of the SDN technique, the control plane network functions (e.g., home subscribe server (HSS), AMF) can be deployed on the public network and shared by the UAV network. The public network and UAV network can maintain their own user plane network functions which are responsible for handling user-specific and bearer traffic.

Based on the above scenario and network architecture, we next mathematically model the communication task.

\subsection{UAV caching model}
As terrestrial users in general submit their content requests earlier than the expected time that content files are received, their content requests may be known in advance \cite{ma2020age}.
For any user $i \in {\mathcal I}$, denote by ${\rm Pr}_i$ the probability of content requests of user $i$ received by the content server with $\sum\nolimits_{i=1}^{N} {\rm Pr}_i = 1$. 
According to the known content request information, decisions on UAV caching including the content placement and delivery, modelled as below, can then be made. 

\subsubsection{Content placement model}
In the model, any terrestrial user may obtain its requested content files from one of the following two links, i.e., UAV cache-user link (if content files are cached in UAVs) and BS cache-UAV-user link\footnote{As the case of obtaining content files via the BS cache-user link was well studied in \cite{ji2020joint,chai2020online,kalantari2020wireless,Chen2017Caching}, we did not investigate this type of link here.}. 
Denote by ${\mathcal F} = \{f_1, f_2, \ldots, f_i, \ldots, f_N\}$ a finite content library, where $f_i$ represents the content file requested by user $i \in {\mathcal I}$. The size of each content file is assumed to be same\footnote{Note that different users may request the similar content file. For files with diverse sizes, the analysis can be extended by dividing each file into chunks of equal size.} \cite{wu2020optimal}.
For any UAV $j \in {\mathcal J}$, denote by $b_{j, f_i}(t)$ a content placement decision variable at time slot $t$.
$b_{j,f_i}(t) = 1$ if the content file $f_i$ destined to user $i$ is cached in UAV $j$. In this case, it is possible for user $i$ to obtain the content file $f_i$ from UAV $j$ directly; otherwise, $b_{j,f_i}(t) = 0$. When $b_{j,f_i}(t) = 0$, user $i$ has to obtain the content file $f_i$ via the BS cache-UAV-user link.
Just like \cite{Chen2017Caching}, we assume that one user requests at most one content file and each UAV can cache at most one content file at time slot $t$. Then, we have the following content placement constraint
\begin{equation}\label{eq:content_saving_indicator}
\sum\nolimits_{f_i \in {\mathcal F}} {b_{j,f_i}(t)}  \le 1, \forall j, t.
\end{equation}

\subsubsection{Content delivery model}
Denote by $s_{ij}(t)$, $\forall i$, $j$, $t$, the content delivery status of UAV $j$ at time slot $t$. $s_{ij}(t) = 1$ represents that UAV $j$ directly delivers the cached content file or forward the content file from the BS cache to user $i$ at $t$. $s_{ij}(t) = 0$ indicates that UAV $j$ does not deliver a content file to user $i$ at $t$. Besides, each UAV can deliver a content file to one user, and one user can obtain the requested content file from one UAV at each time slot. Formally, we have
\begin{equation}\label{eq:slice_request_indicator}
0 \le \sum\nolimits_{j} {{s_{ij}}(t)}  \le 1, \forall i,t, \text{ } 0 \le \sum\nolimits_{i}{{{s_{ij}}(t)}}  \le 1, \forall j,t
\end{equation}
where we lighten $\sum\nolimits_{j}{s_{ij}(t)}$ and $\sum\nolimits_{i}{{{s_{ij}}(t)}}$ for $\sum\nolimits_{j \in {\mathcal J}}{{{s_{ij}}(t)}}$ and ${\sum\nolimits_{i \in {\mathcal I}}{s_{ij}(t)}}$, respectively.
The similar lightened notation is adopted for brevity throughout the rest of this paper.

\subsection{UAV power consumption and movement model}
By referring to (\ref{eq:slice_request_indicator}), we know that a terrestrial user cannot receive content files from a UAV at every time slot. Therefore, we investigate the time average communication behaviors of terrestrial users and UAVs in this paper.

\subsubsection{UAV power consumption model}
Given UAV $j$, define its time average transmit power during the first $t$ time slots as ${\bar p_j}(t) = \frac{1}{t}\sum\nolimits_{\tau  = 1}^t {{p_j}(\tau )}$ with $p_j(\tau)$ being the instantaneous transmit power of UAV $j$ at time slot $\tau$.
Except for the transmit power, UAVs are subject to inherent circuit power consumption mainly including power consumption of mixers, frequency synthesizers, and digital-to-analog converters. Denote $p_{j}^{c}$ as the circuit power of UAV $j$ during a time slot, we then model the power consumption of UAV $j$ at time slot $t$ as
\begin{equation}\label{eq:power_total_at_t}
p_j^{tot}(t) = {p_j}(t) + p_j^c, \forall j, t,
\end{equation}
which is upper-bounded by a constant ${\hat p_j}$, i.e., $p_j^{tot}(t) \le {\hat p_j}$.
Accordingly, the time average power consumption of UAV $j$ during the first $t$ time slots can be written as
\begin{equation}\label{eq:8}
\bar p_j^{tot}(t) = {\bar p_j}(t) + p_j^c, \forall j, t,
\end{equation}
which is constrained by $\bar p_j^{tot}(t) \le {\tilde p_j}$, and ${\tilde p_j}$ is a constant.




\subsubsection{UAV movement model}
At each time slot, all UAVs are movement controlled to execute the communication task efficiently.
Denote the horizontal location of UAV $j \in {\mathcal J}$ as ${\bm{x}}_{j}(t) = [x_{j}(t), y_{j}(t)]^{\rm T} $ at time slot $t$.
Like \cite{wu2020optimal,ji2020joint}, we consider a scenario that all UAVs fly horizontally at a constant altitude $g$ to achieve a lower level of energy consumption.
During the flight, the distance between two consecutive waypoints on a UAV trajectory will be constrained by the UAV's maximum speed. As such, the mathematical expression of the waypoint distance constraint is given by
\begin{equation}\label{eq:waypoint_constr}
 {|| {{{\bm x}_j}(t) - {{\bm x}_j}(t - 1)} ||^2} \le e_{\max }^2, \forall j,t,
\end{equation}
where $e_{\max }$ is the UAV's maximum flight distance during a slot, {${{\bm x}_{j}}(0)$ represents the initial location of $j$}.

Additionally, for collision avoidance, the distance between any two UAVs at each slot should not be less than a safety distance. Mathematically, the expression can be written as
\begin{equation}\label{eq:safety_distance}
{|| {{{\bm x }_j}(t) - {{\bm x}_k}(t)} ||^2} \ge d_{\min }^2, \forall j,k \ne j,t,
\end{equation}
where $d_{\min }$ is the minimum safety distance.

\subsection{Air-to-ground communications}
For air-to-ground (AtG) communications, each terrestrial user may have a line-of-sight (LoS) view towards a UAV with a certain probability. 
A widely adopted expression of the LoS probability is \cite{Al2014Optimal}
\begin{equation}\label{eq:1}
{\rm Pr}({r_{ij}(t)}) = [{{1 + aexp( - b(\phi  _{ij}(t)  - a))}}]^{-1}, \forall i, j, t,
\end{equation}
where $a$ and $b$ are constants relying on the type of environment, such as rural and dense urban, $\phi_{ij}(t)  = {{180} \over \pi } \times \arctan ({{g} \over {{r_{ij}(t)}}})$ is the elevation angle of user $i$ towards UAV $j$, $r_{ij}(t)$ denotes the horizontal distance between user $i \in {\mathcal I}$ and UAV $j \in {\mathcal J}$, i.e., $r_{ij}(t)={||{{\bm x}_j}(t) - {{\bm x}_{i}^{\rm u}(t)}|{|}}$, and ${\bm{x}}_{i}^{\rm u}(t) = [x_{i}^{\rm u}(t), y_{i}^{\rm u}(t)]^{\rm T}$ represents the location of user $i$ at time slot $t$, which can be known via a global positioning system (GPS).

Consider the setting that the user-altitude and antenna-heights of both user and UAV are neglected. The path-loss expression between user $i$ and UAV $j$ can be given by \cite{Al2014Optimal}
\begin{equation}\label{eq:2}
\begin{array}{l}
10\log_{10}(h_{ij}(t)) = 20\log_{10} \left( {\frac{{4\pi}}{\varsigma}} \right) + 20\log_{10} ( {\sqrt {{g^2} + {r_{ij}^2(t)}} } )\\
\qquad \qquad {\rm{ + }} {\rm Pr}({{r_{ij}{(t)}}}){\eta _{LoS}} + (1 - {\rm Pr}({r_{ij}{(t)}})){\eta _{NLoS}}, \forall i, j, t,
\end{array}
\end{equation}
where $\varsigma =c/{{f}_{c}}$ is the carrier wavelength, $c$ (in m/s) is the speed of light, $f_c$ (in Hz) is the carrier frequency, ${\eta _{LoS}}$ (in dB) and ${\eta _{NLoS}}$ (in dB) are losses corresponding to LoS and non line-of-sight (NLoS) connections.

In the considered UAV network, UAVs will {hover} during a time slot to deliver content files to terrestrial users. In this case, it is significant for UAVs to establish LoS links towards terrestrial users. This is because UAVs are energy-constrained and will consume much less power for content provision via LoS links than NLoS links \cite{ji2020joint,ji2020Design}. Therefore, we focus on the LoS AtG communications and then discuss the condition for establishing LoS connections between UAVs and users.
According to statistical analysis results in \cite{Al2014Optimal}, under the worst-case AtG propagation environment (i.e., dense urban), the probability of a LoS AtG propagation link can be over 90\%, when the elevation angle between a UAV and a user is not less than a threshold $\theta^{\rm th}$. Thus, we have the following condition for approximately establishing LoS AtG connections
\begin{equation}\label{eq:horizontal_dist_constraint}
{||{{\bm x}_j}(t) - {{\bm x}_{i}^{\rm u}(t)}|{|}} \le g \tan ^{-1} \theta^{\rm th}, \forall i, j, t.
\end{equation}

Under the approximated LoS AtG connection condition, we can approximate (\ref{eq:2}) as
\begin{equation}\label{eq:los_pathloss}
{h_{ij}}(t) \approx {{G_{LoS}{\varsigma ^2}}}/({{16{\pi ^2}{{( {{{{D_{ij}(t)}}}} )}^2}}}), \forall i, j, t,
\end{equation}
where $G_{LoS} = 10^{-\eta_{LoS}/10}$, and ${{D}_{ij}(t)}=\sqrt {{{g^2} + r_{ij}^2(t)}}$ is the distance between UAV $j$ and user $i$ at time slot $t$.

\subsection{Quality-of-experience model}
In this subsection, the concept of QoE is introduced to understand and improve the subjective perception of the quality of a network service as a whole by the end user \cite{mitra2013context}.
We leverage a widely used mean opinion score (MOS) \cite{mitra2013context} to model the QoE of a user.

\subsubsection{Mean opinion score}
For any user $i$ at time slot $t$, its MOS can take the following form \cite{mitra2013context}
\begin{equation}\label{eq:delay_mos_model}
{\bar D_{i,f_i}(t)} = {{\hat D - {D_{i,f_i}(t)}}}/({{\hat D - L/u_{\rm dl}^{\max }}}), \forall i,t,
\end{equation}
where $\hat D$ is configured based on the desired system requirement \cite{mitra2013context,Chen2017Caching}, ${D_{i,f_i}(t)}$ represents the edge transmission latency of user $i$, which is defined as the required time to transmit a content file from the BS or a UAV to user $i$ at time slot $t$, $u_{\rm dl}^{\max } = {\log _2}\left( {1 + \frac{{{p_{\max } {G_{LoS}\varsigma^2}}}}{16 \pi^2 g^2 \sigma^2W}} \right)$ (in bps/Hz), the constant $p_{\max} = \hat{p}_j-p_j^c $, $\forall j$, is the maximum instantaneous transmit power of a UAV during each time slot, and $\sigma^2W$ is the noise power with $W$ (in MHz) being the total bandwidth.

User $i$ will have a ``very good" QoE state if ${\bar D_{i,f_i}(t)} \ge D^{\rm th}$ \cite{mitra2013context,Chen2017Caching}. Here, $D^{\rm th}$ is the MOS threshold that maximizes the edge transmission latency of user $i$ enjoying the desired ``very good" QoE.
(\ref{eq:delay_mos_model}) shows the interplay between the MOS and the edge transmission latency of user $i$ at time slot $t$. We next model the edge transmission latency of user $i$, $\forall i \in {\mathcal I}$.

\subsubsection{Edge transmission latency model}
Recall that content files can be transmitted to terrestrial users via a BS cache-UAV-user link or a UAV cache-user link. Thus, during time slot $t$, the edge transmission latency for user $i$ to receive a content file $f_i$ can be given by
\begin{equation}\label{eq:latency_definition}
{D_{i,f_i}(t)} = \left\{ {\begin{array}{*{20}{c}}
{{L}/({Wu_{i}^{\rm dl}(t)}),}&{{b_{j,f_i}(t)} = 1}\\
{ {D^{\rm ul}} + {L}/({Wu_{i}^{\rm dl}(t)}),}&{{b_{j,f_i}(t)} = 0}
\end{array}} \right.
\end{equation}
where $L$ (in Mbits) denotes the size of the transmitted content file, $u_{i}^{\rm dl}(t)$ (in bps/Hz) is the achievable data rate of user $i$, and the constant $D^{\rm ul}$ represents the transmission latency of delivering the content file from the BS to a UAV\footnote{The constant transmission latency can be achieved by exploring BS power control strategy, which has been investigated in our previous paper \cite{xi2020network}.}.

From (\ref{eq:delay_mos_model}), we know that improving users' QoE indicates reducing the edge transmission latency of delivering content files.
There is an inherent interplay between the edge transmission latency and data rate.
We next model the data rate.

\subsubsection{AtG data rate model}
By referring to (\ref{eq:delay_mos_model}) and (\ref{eq:latency_definition}), for any UAV $j$ and user $i$, the required data rate (denoted by $C_{i,f_i}^{\rm th}(t)$ (in bps/Hz)) for user $i$ to achieve the desired QoE state when receiving the content file $f_i$ can be written as
\begin{equation}\label{eq:delay_mos_model_upd}
\begin{array}{l}
C_{i,f_i}^{\rm th}(t) = \frac{{L\sum\nolimits_{j \in {{\mathcal J}_i}} {{b_{j,{f_i}}}(t)} }}{{W\left( {{\hat D} - {D^{{\rm{th}}}}({\hat D} - L/(Wu_{{\rm{dl}}}^{\max }))} \right)}} + \\
\frac{{L(1 - \sum\nolimits_{j \in {{\mathcal J}_i}} {{b_{j,{f_i}}}(t)} )}}{{W\left( {{\hat D} - {D^{{\rm{th}}}}({\hat D} - L/(Wu_{{\rm{dl}}}^{\max })) - {D^{{\rm{ul}}}}} \right)}}, \forall i, t,
\end{array}
\end{equation}
where ${\mathcal J}_i(t)$ denotes a set of UAVs with horizontal distances towards user $i$ being small than $e_{\max}$ at time slot $t$.

Next, we use Shannon capacity to quantify the receiving data rate (in bps/Hz) of user $i$ from UAVs at time slot $t$, i.e.,
\begin{equation}\label{eq:achievable_data_rate}
u_{i}^{\rm dl}(t) = \sum\nolimits_{j} {{s_{ij}}(t){{\log }_2}( {1 + \frac{{{p_j}(t){h_{ij}}(t)}}{{{\sigma ^2}W + {I_{ij}}(t)}}} )}, \forall i, t,
\end{equation}
where ${I_{ij}}(t) = \sum\nolimits_{k \in {\mathcal J}\backslash \{ {j}\} } {{p_k}(t){h_{ik}}(t)} $ denotes the interference experienced at user $i$ when all UAVs share the spectrum. Then, the condition of $u_{i}^{\rm dl}(t) \ge C_{i,f_i}^{\rm th}(t)$ should be satisfied if user $i$'s achievable data rate can enable its desired QoE state.

\subsection{PAoI evolution model of packets}
Different from AoI, PAoI provides information of the maximum value of AoI and can capture the extent to which update information is stale \cite{xu2019optimizing}.
Therefore, just like \cite{xu2019optimizing}, we adopt the PAoI as the metric to estimate the freshness of information.
Additionally, we discuss the above models from the perspective of content files, each of which includes a batch of packets.
However, as most PAoI-related work \cite{xu2019optimizing,talak2019optimizing} discuss the PAoI evolution from the viewpoint of data packets, we next model the PAoI evolution of packets.

\begin{figure}[!t]
\centering
\begin{minipage}[t]{0.24\textwidth}
\centering
\includegraphics[width=1.7 in]{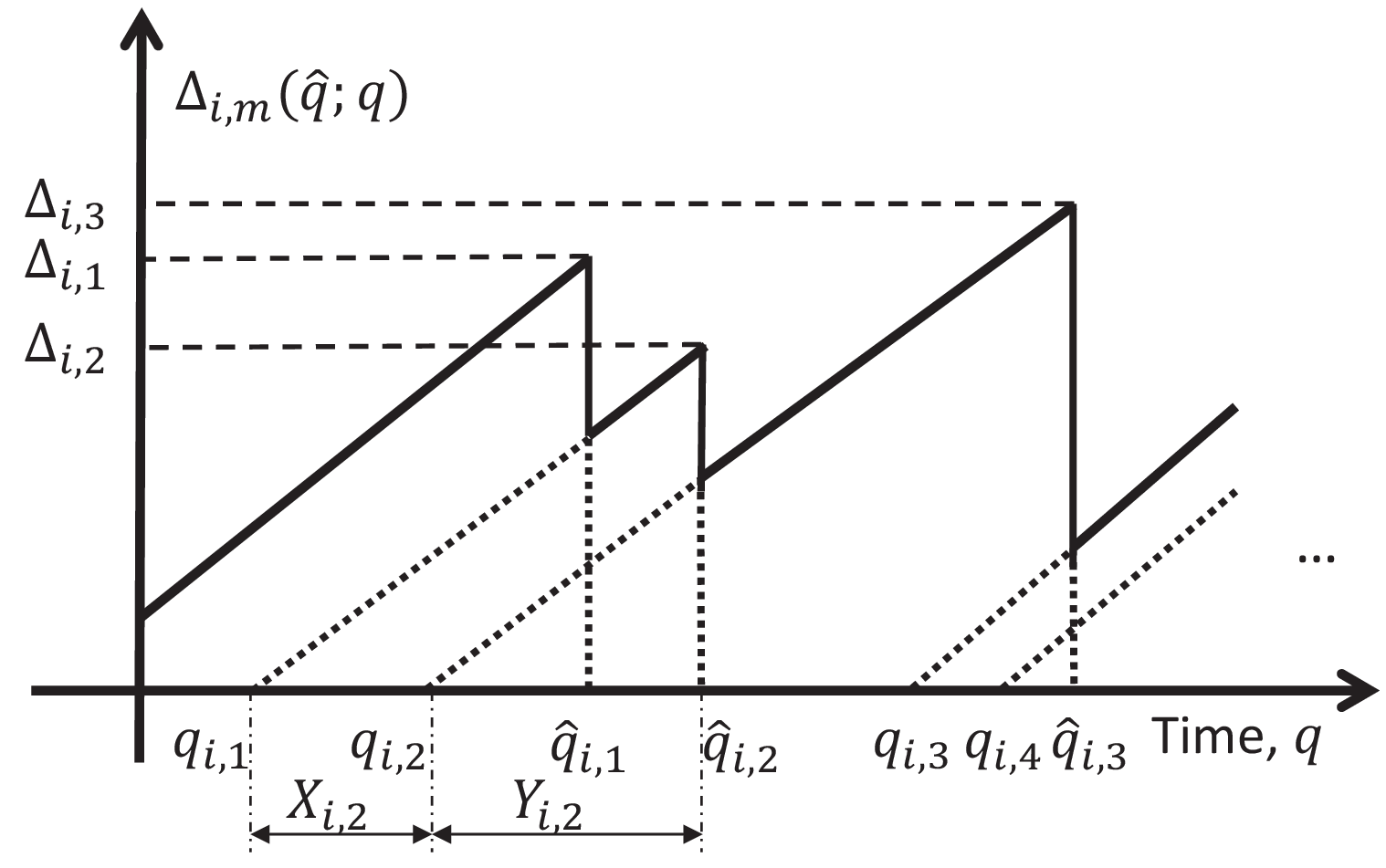}
\caption{An example of PAoI evolution model of packets towards user $i$ \cite{xu2019optimizing}.}
\label{fig_PAoI_model}
\end{minipage}
\begin{minipage}[t]{0.24\textwidth}
\centering
\includegraphics[width=1.7 in]{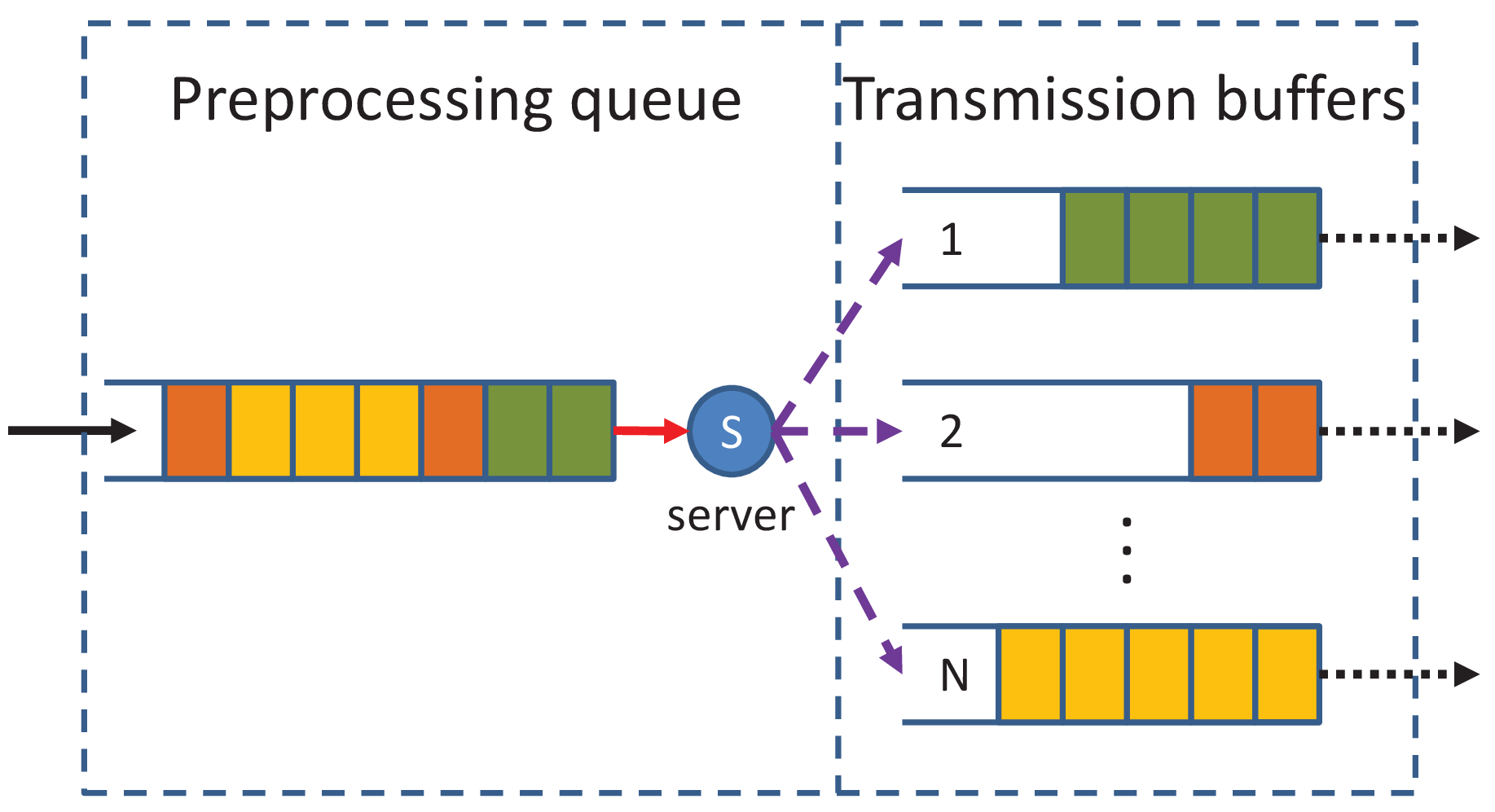}
\caption{Management strategy of the BS for newly arrived packets.}
\label{fig_queue_model}
\end{minipage}
\end{figure}
For any user $i \in {\mathcal I}$, define the AoI of a packet $m$ destined to user $i$ as $\Gamma _{i,m}(q) = q - {q_{i,m-1}}$, where ${q_{i,m-1}}$ is the generation time of the most recently received packet $m-1$ from the data server until time $q$ \cite{kaul2012real}.\footnote{{The timescale of modelling of PAoI differs from that in the above models. This is because the inter-arrival time of packets is different from that of transmitting content files to users. For clarification, we use the terminology `time' and the corresponding notation `q' when analyzing PAoI of pcakets.}}
When user $i$ does not receive packet $m$, the value of $\Gamma_{i,m}(q)$ increases linearly with $q$, which in turn shows the fact that packet $m$ is getting older.
In other words, the $m$-th peak value of $\Gamma_{i,m}(q)$ is obtained right before the $m$-th newly generated packet arrives at user $i$. The $m$-th peak value of $\Gamma_{i,m}(q)$ is defined as the PAoI \cite{xu2019optimizing}, denoted by $\Delta_{i,m}(\hat q; q)$, where $\hat q$ denotes the time that a packet generated at time $q$ arrives at its destination user.
Fig. \ref{fig_PAoI_model} shows an example of the PAoI evolution model $\Delta_{i,m}(\hat q;q)$ for user $i$ and packet $m$.\footnote{Like \cite{xu2019optimizing,talak2019optimizing}, we discuss the PAoI evolution of packets in the discrete time domain.}
Formally, for any user $i$ and packet $m$, the PAoI of packet $m$ destined to user $i$ evolves as follows \cite{xu2019optimizing}
\begin{equation}\label{eq:AoI_model}
{\Delta_{i,m}(\hat q; q)} = \left\{ {\begin{array}{*{20}{l}}
{{{\hat q}_{i,m}}}&{m = 1}\\
{{X_{i,m}} + {Y_{i,m}}}&{m > 1}
\end{array}} \right.
\end{equation}
where $\hat q_{i,m}$ is the time that packet $m$ reaches user $i$,
$X_{i,m} = q_{i,m} - q_{i,m-1}$, $Y_{i,m} = \hat q_{i,m} - q_{i,m}$, and $q_{i,m}$ is the generation time of packet $m$ towards user $i$, as shown in Fig. \ref{fig_PAoI_model}.

The inter-arrival time $X_{i,m}$ is related to the packet arrival rate of packets sent to user $i$.
The value of $Y_{i,m}$ is determined by many factors, such as the time cost of preprocessing data packets in the BS and the edge arrival duration.
Specifically, for the BS, it will maintain a queue with infinite buffer space to manage arrived data packets, as shown in Fig. \ref{fig_queue_model}. Upon arriving at the BS, data packets will enter into the BS queue to wait to be preprocessed (e.g., packet classification, packet header update) according to a first-come-first-served (FCFS) principle. We call this queue the preprocessing queue. After being preprocessed, packets destined to different users will be cached at diverse transmission buffers maintained by the BS or forwarded and cached in UAVs.
Therefore, we can re-write $Y_{i,m}$ as $Y_{i,m} = Y_{i,m}^{\rm Q} + Y_{i,m}^{\rm S} + Y_{i,m}^{\rm A}$, where $Y_{i,m}^{\rm Q}$ denotes the queueing delay of packet $m$, $Y_{i,m}^{\rm S}$ is the preprocessing time of packet $m$ in the content file $f_i$, and $Y_{i,m}^{\rm A} \le D_{i,f_i}(t)$ is the edge arrival duration of packet $m$.
It is noteworthy that the values of $X_{i,m}$, $Y_{i,m}^{\rm Q}$, and $Y_{i,m}^{\rm S}$ are determined by some network parameters such as backhaul capacity and CPU computing speed and cannot be reduced by optimizing the deployment and resource (e.g., UAV transmit power) allocation of the UAV network \cite{xu2019optimizing}; yet, $Y_{i,m}^{\rm A}$ can be optimized.
In the following subsection we formulate an optimization problem of reducing the value of $Y_{i,m}^{\rm A}$ while providing fair and energy-efficient content file delivery for all terrestrial users via the joint design of UAV caching, UAV trajectory, and UAV transmit power.

\subsection{Problem formulation}
Improving users' achievable data rates will result in reduced edge arrival duration and the following fresher content file transmission. Thus, a goal of the problem is to maximize users' achievable data rates.
During the first $t$ time slots, the time average achievable data rate of user $i$, $\forall i$, is written as ${\bar u_{i}^{\rm dl}}(t) = \frac{1}{t}\sum\nolimits_{\tau  = 1}^t {{u_{i}^{\rm dl}}(\tau )}$. Define $\phi ( \{{{\bar{u}}_{i}^{\rm dl}}(t)\} )={\sum\nolimits_{i}{{{\log }_{2}}(1+{{{\bar{u}}}_{i}^{\rm dl}}}(t))}$ as a proportional fairness function of time average achievable data rates across all terrestrial users.
Then, maximizing $\phi (\{{{\bar{u}}_{i}^{\rm dl}}(t)\})$ will result in fresh and fair content provision for all users.
Besides, to implement the energy-efficient content file delivery, the power consumption of UAVs should be minimized.
To achieve the above goals, the joint optimization of UAV caching, UAV trajectory, and UAV transmit power should be investigated. Mathematically, we can formulate the joint optimization problem as below
\begin{subequations}\label{eq:original_problem}
\begin{alignat}{2}
& \mathop {\rm Maximize }\limits_{{\mathcal B}(t),{\mathcal S}(t),{\mathcal P}(t),{{{\mathcal X}}}(t)} {\mkern 1mu} \mathop {\lim \inf }\limits_{t \to \infty } ( {\phi ( \{{\bar u_{i}^{\rm dl}}(t)\} ) - \rho \sum\nolimits_{j} {\bar p_j^{tot}(t)} } ) \allowdisplaybreaks[4] \\
&{\rm s.t.} \text{ } \mathop {\lim \inf }\limits_{t \to \infty } [{\bar u_{i}^{\rm dl}}(t)-\bar C_{i,f_i}^{\rm th}(t)] \ge 0,\forall i \allowdisplaybreaks[4] \\
& \quad \mathop {\lim \sup }\limits_{t \to \infty } \bar p_j^{tot}(t) \le {{\tilde p}_j},\forall j \\
& \quad  p_j^{tot}(t) \le {{\hat p}_j},\forall j,t \\
& \quad {s_{ij}}(t) \in \{ 0,1\} ,\forall i,j,t \\
& \quad {b_{j,f_i}}(t) \in \{ 0,1\} ,\forall j,t \\
& \quad {p_j}(t) \ge p_j^{\min},{\rm{ }}\forall j,t \\
& \quad (\ref{eq:content_saving_indicator}),(\ref{eq:slice_request_indicator}),(\ref{eq:waypoint_constr}),
(\ref{eq:safety_distance}),(\ref{eq:horizontal_dist_constraint}).
\end{alignat}
\end{subequations}
where ${\mathcal B}(t)$, ${\mathcal S}(t)$, ${\mathcal P}(t)$, and ${{{\mathcal X}}}(t)$ represent the sets of content placement decision variables, content delivery decision variables, UAV transmit power, and UAV locations at time slot $t$, respectively, $\rho$ is a non-negative coefficient that weighs the trade-off between fresh and fair content delivery and power consumption, $p_j^{\min}$ is a small constant, $\bar C_{i,f_i}^{\rm th}(t ) = \frac{1}{t}\sum\nolimits_{\tau  = 1}^t {\varphi {C_{i,f_i}^{\rm th}}(\tau )}$ with $\varphi = \frac{J}{N}$. The constant $\varphi$ is introduced because the condition $u_{i}^{\rm dl}(t) \ge C_{i,f_i}^{\rm th}(t)$ should be satisfied if user $i$'s ``very good" QoE state can be achieved at time slot $t$. However, user $i$ cannot receive a content file from a UAV at each time slot. Besides, each user has the probability of $\frac{J}{N}$ to receive the content file from a UAV due to the goal of achieving fair content delivery.

The solution to (\ref{eq:original_problem}) is quite challenging mainly because i) \emph{time-coupled objective function:} the objective function is the logarithmic function of time average achievable data rates. The calculation of the objective function requires the obtaining of all users' achievable data rates over the first $t$ time slots, which indicates the optimization of a great number of decision variables. Besides, the number of decision variables in the problem will exponentially increase with an increasing $t$, which seriously hinders the solution to the problem;
ii) \emph{sequential decision problem:} it needs to optimize the UAV cache placement scheme, UAV cache delivery strategy, UAV transmit power, and UAV trajectories during the first $t$ time slots;
iii) \emph{thorny optimization problem:} (\ref{eq:original_problem}) includes a logarithmic-quadratic objective function, non-convex constraints (explained in detail in Section III), and continuous and integer variables. Therefore, (\ref{eq:original_problem}) is a mixed-integer non-convex optimization problem that may be NP-hard or even undecidable.

To solve this highly challenging problem, we propose a Lyapunov-based optimization framework. In this framework, we first attempt to decouple the objective function of (\ref{eq:original_problem}) in terms of time slots.
Next, we leverage a Lyapunov drift-plus-penalty technique \cite{Neely2014A} to further decompose the problem with a time-decoupled objective function into multiple repeatedly optimized subproblems. Finally, an iterative optimization scheme is designed to tackle the mixed-integer non-convex characteristic of the subproblems.


\section{Lyapunov-Based optimization Framework}
Observing that the objective function is the logarithmic function of time average achievable data rates, we refer to the objective function as time-coupled objective function. We first leverage the Jensen inequality to decouple the time-coupled objective function.
A sequential decision problem with a time-decoupled objective function can then be obtained, which is still difficult to be solved effectively.
Reinforcement learning (RL) approaches, such as Q-learning \cite{watkins1992q}, and deep deterministic policy gradient (DDPG) \cite{lillicrap2015continuous}, can be explored to solve sequential decision problems.
However, Q-learning-based approaches can only handle discrete and low-dimensional action spaces, and DDPG-based methods are designed for continuous (real valued) action spaces \cite{lillicrap2015continuous}. (\ref{eq:original_problem}) simultaneously involves discrete and continuous action spaces, which indicates that it will be highly difficult to design RL approaches to solve (\ref{eq:original_problem}). What's more, RL approaches suffer from lack of complete theoretical basis.

To solve the sequential decision problem effectively, we propose a Lyapunov-based optimization framework, which decomposes the problem into mutiple repeatedly optimized subproblems rather than solve this problem as a whole. The procedure of the framework is as follows.

\subsection{Decouple of the objective function}
Let ${\bm \gamma} (t)=({{\gamma }_{1}}(t),\ldots ,{{\gamma }_{N}}(t))$ be an auxiliary vector with {$0\le {{\gamma }_{i}}(t)\le u_{\rm dl}^{\max }$, $\forall i, t$}. Define $g(t)=\phi (\bm{\gamma} (t)) = {\sum\nolimits_{i }{\log_2(1 + \gamma_{i}(t))}}$. Then, according to the Jensen's inequality, we can achieve $\bar{g}(t)\le \phi ({{\bar{\bm \gamma }}}(t))$. With this important inequality, the following Proposition shows that we can equivalently transform the original problem into a new one with a time-decoupled objective function.
\begin{prop}\label{prop:time_decouple}
{\rm The original problem (\ref{eq:original_problem}) can be equivalently transformed into the following sequential decision problem.
\begin{subequations}\label{eq:Jensen_problem}
\begin{alignat}{2}
& \mathop {\rm Maximize }\limits_{{\mathcal B}(t),{\mathcal S}(t),{\mathcal P}(t),{{{\mathcal X}}}(t),{\bm {\gamma} (t)}} {\mkern 1mu} \text{ } \mathop {\lim \inf }\limits_{t \to \infty } ( {\bar g(t) - \rho \sum\nolimits_{j } {{{{\bar p}_j^{tot}}(t)}} } ) \allowdisplaybreaks[4] \\
&{\rm s.t.} \quad {\mathop {\lim \inf }\limits_{t \to \infty } [{{\bar u}_{i}^{\rm dl}}(t) - {{\bar \gamma }_{i}}(t)] = 0,\forall i} \allowdisplaybreaks[4] \\
& \quad \mathop {\lim \inf }\limits_{t \to \infty } [{{\bar u}_{i}^{\rm dl}}(t) - {\bar C_{i,f_i}^{\rm th}(t )}] \ge 0,\forall i \allowdisplaybreaks[4] \\
& \quad \mathop {\lim \inf }\limits_{t \to \infty } [{{\tilde p}_j} - {{\bar p}_j^{tot}}(t)] \ge 0,\forall j \\
& \quad 0 \le {\gamma _{i}}(t) \le u_{i}^{\max }, \forall i,t \\
& \quad {\rm (\ref{eq:original_problem}d)-(\ref{eq:original_problem}h).}
\end{alignat}
\end{subequations}
}
\end{prop}
\begin{proof}
Please refer to Appendix A.
\end{proof}

\subsection{Lyapunov drift-plus-penalty}
In this subsection, we leverage a Lyapunov drift-plus-penalty technique \cite{Neely2014A} to tackle the time average constraints in (\ref{eq:Jensen_problem}).
Specifically, to enforce the constraint (\ref{eq:Jensen_problem}c), we introduce a family of virtual queues $\{{{Q}_{i}}(t)\}$ as the following
\begin{equation}\label{eq:Queue_ui}
{Q_{i}}(t ) = {Q_{i}}(t-1) + {\varphi C_{i,f_i}^{\rm th}(t)} - {u_{i}^{\rm dl}}(t-1), \forall i,t.
\end{equation}

It can be concluded that the constraint (\ref{eq:Jensen_problem}c) is satisfied if the following mean-rate stability condition holds \cite{Neely2014A}
\begin{equation}\label{eq:Queue_EQ}
\mathop {\lim }\nolimits_{t \to \infty } {{{\mathbb E}\{{{[{Q_{i}}(t)]}^ + }\}}}/{t} = 0, \forall i,
\end{equation}
where the non-negative operation ${[x]^ + } = \max \{ x,0\}$.

Likewise, to enforce the time average constraints (\ref{eq:Jensen_problem}b) and (\ref{eq:Jensen_problem}d), we define the virtual queues $Z_{i}(t)$, and ${{H}_{j}}(t)$, respectively, as
\begin{equation}\label{eq:Queue_Z}
{Z_{i}}(t ) = {Z_{i}}(t-1) + {\gamma _{i}}(t-1) - {u_{i}^{\rm dl}}(t-1), \forall i, t,
\end{equation}
\begin{equation}\label{eq:Queue_H}
{H_j}(t ) = {H_j}(t-1) + p_j^{tot}(t-1) - {\tilde p_j}, \forall  j, t.
\end{equation}

(\ref{eq:Jensen_problem}b) and (\ref{eq:Jensen_problem}d) can be satisfied, if the following mean-rate stability conditions can be held
\begin{equation}\label{eq:Queue_EZ}
\mathop {\lim }\nolimits_{t \to \infty } {{{\mathbb E}\{{{[{Z_{i}}(t)]}^ + }\}}}/{t} = 0, \forall i,
\end{equation}
\begin{equation}\label{eq:Queue_EH}
\mathop {\lim }\nolimits_{t \to \infty } {{{\mathbb E}\{{{[{H_j}(t)]}^ + }\}}}/{t} = 0, \forall j.
\end{equation}

With the definitions of the virtual queues ${{[{{Q}_{i}}(t)]}^{+}}$, ${{[{{Z}_{i}}(t)]}^{+}}$, and ${{[{{H}_{j}}(t)]}^{+}}$, we can define a Lyapunov function $L\left( t \right)$ as a sum of square of these virtual queues at time slot $t$, i.e.,
${L(t)} \text{ } { \buildrel \Delta \over = \text{ } \frac{1}{2}\sum\limits_{i} {{{({{[{Q_{i}}(t)]}^ + })}^2}}  + \frac{1}{2}\sum\limits_{i} {{{({{[{Z_{i}}(t)]}^ + })}^2}}  }$ $ + {\frac{1}{2}\sum\limits_{j } {{{({{[{H_j}(t)]}^ + })}^2}} }$.
$L(t)$ is a scalar measure of constraint violations. Intuitively, if the value of $L(t)$ is small, the absolute values of all queues are small; otherwise, the absolute value of at least one queue is great.
Additionally, we define a drift-plus-penalty function as $\Delta (t)-V\left( g(t)-\rho \sum\nolimits_{j }{ {{p}_{j}^{tot}}(t)} \right)$, where $\Delta (t) = L(t + 1) - L(t)$ represents a \emph{Lyapunov drift}, $-\left( g(t)-\rho \sum\nolimits_{j }{{{p}_{j}^{tot}}(t)} \right)$ is a \emph{penalty}, and $V$ is a non-negative penalty coefficient that weighs the trade-off between constraint violations and optimality.
Lemma \ref{lemma:1} presents the upper bound of the function value.

\begin{lemma}\label{lemma:1}
{\rm At each time slot $t$, the upper bound of the value of the drift-plus-penalty function $\Delta (t)-V\left( g(t)-\rho \sum\nolimits_{j }{{{p}_{j}^{tot}}(t)} \right)$ can be expressed as (\ref{eq:upper_bound}) with $B{\rm{ }} \buildrel \Delta \over = {\sum\nolimits_{i} {{{(u_{\rm dl}^{\max })}^2}}}  + \sum\nolimits_{j } {{{(\hat p_j)}^2}/2}$
\begin{equation}\label{eq:upper_bound}
\begin{array}{l}
\Delta (t) - V\left( {g(t) - \rho \sum\nolimits_j {p_j^{tot}(t)} } \right) \le B - \\
\quad \sum\nolimits_j {{{[{H_j}(t)]}^ + }\left( {{{\tilde p}_j} - p_j^c} \right)}  + V\rho \sum\nolimits_j {p_j^c}  + \\
\quad \sum\nolimits_i {{{[{Q_i}(t)]}^ + }\varphi C_{i,{f_i}}^{{\rm{th}}}(t)}  - V\phi ({\bm \gamma} (t)) + \\
\quad \sum\nolimits_i {{{[{Z_i}(t)]}^ + }{\gamma _i}(t)}  + \sum\nolimits_j {\{ V\rho  + {{[{H_j}(t)]}^ + }\} {p_j}(t)}  - \\
\quad \sum\nolimits_i {\{ {{[{Q_i}(t)]}^ + } + {{[{Z_i}(t)]}^ + }\} u_i^{{\rm{dl}}}(t)}.
\end{array}
\end{equation}
}
\end{lemma}

\begin{proof}
Please refer to Appendix B.
\end{proof}

In (\ref{eq:upper_bound}), the right-hand-side expression constitutes the upper bound of the drift-plus-penalty. As such, the minimization of the drift-plus-penalty can be approximated by minimizing its upper bound.
We therefore mitigate (\ref{eq:original_problem}) by greedily minimizing the upper bound of the drift-plus-penalty function at each $t$.
Meanwhile, at each $t$, the upper bound can be decomposed into four independent terms including a constant term, an auxiliary variable term, a term related to content caching, and a term consisting of content delivery strategies, UAV transmit power as well as UAV trajectories. As a result, the Lyapunov-based optimization framework of mitigating (\ref{eq:original_problem}) can be summarized as the following repeated optimization subproblems of three-tier structure.
\begin{itemize}
\item At each time slot $t$, observe $Q_{i}(t)$, ${Z_{i}}(t)$, ${{H}_{j}}(t)$ for any user $i \in {\mathcal I}$, and UAV $j \in {\mathcal J}$.
\item \emph{AUxiliary-Tier (AUT) optimization:} Choose ${{\gamma _{i}}(t)}$ for each user $i$ to mitigate (\ref{eq:gamma_related_problem})
\begin{subequations}\label{eq:gamma_related_problem}
\begin{alignat}{2}
& \mathop {\rm Minimize }\limits_{{\bm {\gamma}} (t)} {\mkern 1mu} \text{ } \mathop  - V\phi ({\bm {\gamma} (t)}) + {\sum\nolimits_{i}{{{[{Z_{i}}(t)]}^ + }{\gamma _{i}}(t)}} \\ &{\rm {s.t.} \text{ }} \quad 0 \le \gamma_{i}(t) \le u_{\rm dl}^{\max}
\end{alignat}
\end{subequations}
\item \emph{Content-Placement-Tier (CPT) optimization}: Determine the content placement decision variable $b_{j,f_i}(t)$ for each UAV $j$ to optimize (\ref{eq:caching_placement_problem})
\begin{subequations}\label{eq:caching_placement_problem}
\begin{alignat}{2}
& \mathop {\rm Minimize }\limits_{{\mathcal B}(t)} {\mkern 1mu} \text{ } \mathop  {\sum\nolimits_{i}{{{[{Q_{i}}(t)]}^ + }C_{i,f_i}^{\rm th}(t) }} \\
&{\rm {s.t.} \text{ }}  \quad {\rm constraint} \text{ } (\ref{eq:content_saving_indicator}).
\end{alignat}
\end{subequations}
\item \emph{Delivery-Power-and-Trajectory-Tier (DPT$^2$) optimization}: Given UAV trajectories ${{\mathcal {X}}(t-1)}$, choose ${{\mathcal {S}}(t)}$, ${{\mathcal {P}}(t)}$, and ${{ {{\mathcal X}}}(t)}$ to mitigate (\ref{eq:subproblem_BPX})
\begin{subequations}\label{eq:subproblem_BPX}
\begin{alignat}{2}
& \mathop {\rm Minimize }\limits_{{\mathcal {S}}(t),{\mathcal {P}}(t),{{{\mathcal X}}(t)}} {\mkern 1mu} \text{ }  \sum\nolimits_{j } {\{ V\rho  + {{[{H_j}(t)]}^ + }\} {p_j}(t)}  - \nonumber \\
& \qquad \quad {\sum\nolimits_{i}{\{ {{[{Q_{i}}(t)]}^ + } + {{[{Z_{i}}(t)]}^ + }\} {u_{i}^{\rm dl}}(t)}}   \\
& {\rm s.t.} \quad {\rm (\ref{eq:slice_request_indicator}),(\ref{eq:waypoint_constr}),
(\ref{eq:safety_distance}),(\ref{eq:horizontal_dist_constraint}),(\ref{eq:original_problem}d)-(\ref{eq:original_problem}g).}
\end{alignat}
\end{subequations}
\item Compute $u_{i}^{\rm dl}(t)$ using (\ref{eq:achievable_data_rate}). Update there virtual queues using (\ref{eq:Queue_ui}), (\ref{eq:Queue_Z}), and (\ref{eq:Queue_H}).
\end{itemize}



As shown in the above framework, the solution of (\ref{eq:original_problem}) lies in the optimization of some subproblems. In this section, we present the detailed procedure of solving it.

\subsection{AUT optimization}
As the proportional fairness function $\phi (\bm \gamma (t))$ is a separable sum of individual logarithmic functions, the mitigation of (\ref{eq:gamma_related_problem}) is equivalent to a separate selection of the individual auxiliary variable {${{\gamma }_{i}}(t)\in [ 0,u_{\rm dl}^{\max } ]$} for each user $i \in {\mathcal I}$ that minimizes a convex function $-V{{\log }_{2}}(1+{{\gamma }_{i}}(t))+{{[{{Z}_{i}}(t)]}^{+}}{{\gamma }_{i}}(t)$ with respect to (w.r.t.) $\gamma_i(t)$. Thus, the closed-form solution to (\ref{eq:gamma_related_problem}) can be written as
\begin{equation}\label{eq:compute_gamma}
{\gamma _{i}}(t) = \left\{ {\begin{array}{*{20}{l}}
{u_{\rm dl}^{\max }, \quad {{{[{Z_{i}}(t)]}^ + } = 0}}\\
{\min \left\{ {{{\left[ {\frac{V}{{{{[{Z_{i}}(t)]}^ + }\ln 2}} - 1} \right]}^ + },\;u_{\rm dl}^{\max }} \right\},\text{ } {\rm else}}
\end{array}} \right.
\end{equation}

\subsection{CPT optimization}
The goal of (\ref{eq:caching_placement_problem}) is to reduce the data rate requirements of all terrestrial users. To this aim, content files should be cached in UAVs.
As a result, the total power of all UAVs can be reduced when delivering content files to terrestrial users.
Minimizing the total power of all UAVs is equal to the maximization of the reduction of transmit power of each UAV brought by content caching \cite{Chen2017Caching}. 
Thus, we can design the following content placement scheme to solve (\ref{eq:caching_placement_problem})
\begin{equation}\label{eq:greedy_bk}
{b_{j,{f_i}}}(t) = \left\{ {\begin{array}{*{20}{l}}
{1,}&{i = i^{\star}}\\
{0,}&{\rm otherwise}
\end{array}} \right.
\end{equation}
where
\begin{equation}\label{eq:i_star}
{i^ \star } = \mathop {\arg \max }\limits_{i \in {{\cal N}_j}} {[{Q_i}(t)]^ + }({p_j}(\beta) - {p_j}(\alpha )),
\end{equation}
and ${\mathcal N}_j$ is the set of terrestrial users with the horizontal distance towards UAV $j$ at time slot $t$ being small than $e_{\max}$, $p_j(\beta) - p_j(\alpha)$ represents the transmit power reduction of UAV $j$ due to the content caching. Besides, UAV $j$ will cache the content for its nearest user if ${\mathcal N}_j = \emptyset $. $\beta  = \frac{L}{{\left( {\hat D - {D^{{\rm{th}}}}(\hat D - L/u_{{\rm{dl}}}^{\max }) - {D^{{\rm{ul}}}}} \right)}}$, $\alpha  = \frac{L}{{\left( {\hat D - {D^{{\rm{th}}}}(\hat D - L/u_{{\rm{dl}}}^{\max })} \right)}}$, and $p_j(\varpi ) = \frac{(2^{\varpi/W}-1)(\sigma^2W+I_{ij}(t))}{h_{ij}(t)}$.

{Remark 2:} From (\ref{eq:greedy_bk}) and (\ref{eq:i_star}), we can see that the content placement decision will be made based on the states (e.g., UAV locations) of the UAV network at the current time slot.
Besides, the content placement depends on the pre-knowledge of content request of each terrestrial user and users' locations, which corresponds to the result given in \cite{Chen2017Caching}.

\subsection{DPT$^2$ optimization}
It can be observed that (\ref{eq:subproblem_BPX}) includes logarithmic-quadratic-terms and continuous and integer variables. Besides, the constraint (\ref{eq:original_problem}f) is non-convex; thus, (\ref{eq:subproblem_BPX}) is a mixed-integer non-convex programming problem that is difficult to be addressed directly.

To address this challenge, we first propose to tackle the mixed-integer issue of (\ref{eq:subproblem_BPX}) by leveraging an iterative optimization scheme. Particularly, the solution to (\ref{eq:subproblem_BPX}) includes the iterative optimization of content delivery decision variables, UAV trajectory, and transmit power.
Second, we explore an SCA technique \cite{scutari2016parallel} to approximately convert the generated non-convex optimization problems during the iterative optimization into convex ones.
\subsubsection{Content delivery decision variable optimization}
For given UAV trajectories ${\mathcal X}(t)$ and transmit power ${\mathcal P}(t)$, the content delivery stratety of (\ref{eq:subproblem_BPX}) can be developed by solving the following problem
\begin{subequations}\label{eq:UE_association_problem}
\begin{alignat}{2}
& \mathop {\rm Maximize }\limits_{{\mathcal {S}}(t)} {\mkern 1mu} \text{ } {\sum\nolimits_{i} {\sum\nolimits_{j } {{c_{ij}(t)}{s_{ij}}(t)} } }  \\
&{\rm s.t. \text{ }  \quad} {\rm (\ref{eq:slice_request_indicator}), (\ref{eq:original_problem}e).}
\end{alignat}
\end{subequations}
where ${c_{ij}(t)} = \{ {[{Q_{i}}(t)]^ + } + {[{Z_{i}}(t)]^ + }\}{\log _2}\left( {1 + \frac{{{p_j}(t){h_{ij}}(t)}}{{{\sigma ^2}W + {I_{ij}}(t)}}} \right)$.

It is easy to know that (\ref{eq:UE_association_problem}) is an integer linear programming problem and can be efficiently solved by some optimization tools such as MOSEK.

\subsubsection{UAV trajectory optimization}
Given the UAV transmit power ${\mathcal P}(t)$, UAV trajectories at time slot $t-1$, ${{\mathcal {X}}(t-1)}$, and the content delivery decision variable set ${\mathcal S}(t)$, (\ref{eq:subproblem_BPX}) is still difficult to be solved by some standard optimization methods due to the non-convex objective function and constraint (\ref{eq:safety_distance}). To solve this problem effectively, an SCA technique \cite{scutari2016parallel} is exploited to tackle the non-convexity and approximately transform the non-convex optimization problem into a convex one. The key idea of SCA is to solve a sequence of convex optimization problems with different initial points to obtain an approximate solution to a non-convex optimization problem instead of solving the hard non-convex problem directly.
The following Proposition presents the approximately transformed convex UAV trajectory optimization problem.
\begin{prop}\label{lemma:lemma_uav_location}
{\rm By exploring the SCA technique, UAV trajectories at time slot $t$ can be obtained by mitigating the following convex optimization problem.
\begin{subequations}\label{eq:UAV_location_equal_problem_approximate}
\begin{alignat}{2}
& \mathop {\rm Maximize }\limits_{{\mathcal {X}}(t), \{\eta_{i}(t)\}, \{B_{ik}(t)\}} {\mkern 1mu} \text{ } {{\sum\nolimits_{i} {\{ {{[{Q_{i}}(t)]}^ + } + {{[{Z_{i}}(t)]}^ + }\} {\eta _{i}}(t)} }}  \\
& {\rm s.t.} \text{ } {\sum\limits_{j } {{s_{ij}}(t)( {D_{i}^{(r)}(t) - \sum\limits_{k \in {\mathcal J}} {E_{ik}^{(r)}(t)( {||{{\bm x}_k}(t) - {{\bm x}_{i}^{\rm u}(t)}|{|^2} - } } } } } \nonumber \\
& \ { { {||{\bm x}_k^{(r)}(t) - {{\bm x}_{i}^{\rm u}(t)}|{|^2}} )} ) + \sum\nolimits_{j } {{s_{ij}}(t){{\tilde R}_{ij}}(t)}  \ge {\eta _{i}}(t),\forall i,t } \allowdisplaybreaks[4] \\
& \quad {B_{ik}(t)} \le  {|| {{\bm x}_k^{(r)}(t) - {\bm x}_{i}^{\rm u}(t)} ||^2} + \nonumber \\
& \quad 2{( {{\bm x}_k^{(r)}(t) - {\bm x}_{i}^{\rm u}(t)} )^{\rm T}}\left( {{{\bm x}_k}(t) - {{\bm x}_{i}^{\rm u}(t)}} \right),\forall i,k \ne j,t \\
& \quad - {|| {{\bm x}_j^{(r)}(t) - {\bm x}_k^{(r)}(t)} ||^2} + 2{( {{\bm x}_j^{(r)}(t) - {\bm x}_k^{(r)}(t)} )^{\rm T}} \times \nonumber \\
& \qquad \left( {{{\bm x}_j}(t) - {{\bm x}_k}(t)} \right) \ge d_{\rm min}^2, \forall j,k \ne j,t \\
& \quad {\rm (\ref{eq:waypoint_constr}),(\ref{eq:horizontal_dist_constraint})}
\end{alignat}
\end{subequations}
where $\eta_{i}(t)$ and ${B_{ik}(t)}$ are slack variables, $D_{i}^{(r)}(t) = {\log _2}( {{\sigma ^2}W + \sum\limits_{k \in {\mathcal J}} {\frac{{{{p_k}(t)\theta _{ij}}}}{{{g^2} + ||{\bm x}_k^{(r)}(t) - {{\bm x}_{i}^{\rm u}(t)}|{|^2}}}} } )$, $E_{ik}^{(r)}(t) = {{\frac{{{{p_k}(t)\theta _{ij}}}}{{{{\left( {{g^2} + ||{\bm x}_k^{(r)}(t) - {{\bm x}_{i}^{\rm u}(t)}|{|^2}} \right)}^2}{2^{D_{i}^{(r)}(t)}\ln2}}}}}$, ${\tilde R_{ij}}(t) =  - {\log _2}( {{\sigma ^2}W + \sum\nolimits_{k \in {\mathcal J}\backslash \{ j\} } {\frac{{{{p_k}(t)\theta _{ij}}}}{{{g^2} + {B_{ik}}(t)}}} } )$, ${\theta _{ij}} = \frac{{G_{LoS}{\varsigma ^2}}}{{16{\pi ^2}}}$, ${\bm x}_{j}^{(r)}(t)$, and ${\bm x}_{k}^{(r)}(t)$ are given locations of UAV $j$ and UAV $k$ at the $r$-th iteration of the SCA technique.}
\end{prop}
\begin{proof}
Please refer to Appendix C.
\end{proof}

{{Remark 3}}: The objective function (\ref{eq:UAV_location_equal_problem_approximate}a) is linear. As the left-hand-side (LHS) of (\ref{eq:UAV_location_equal_problem_approximate}b) is concave w.r.t. both $\bm x_k(t)$ and $B_{i,k}(t)$, it is a convex constraint. (\ref{eq:UAV_location_equal_problem_approximate}c) and (\ref{eq:UAV_location_equal_problem_approximate}d) are linear constraints. Besides, both (\ref{eq:horizontal_dist_constraint}) and (\ref{eq:waypoint_constr}) are convex quadratic constraints. Thus, (\ref{eq:UAV_location_equal_problem_approximate}) is now convex and can be efficiently mitigated by MOSEK. It is noteworthy that the lower-bounded approximation conducted in (\ref{eq:UAV_location_equal_problem_approximate}b)-(\ref{eq:UAV_location_equal_problem_approximate}d) shows that the feasible domain of (\ref{eq:UAV_location_equal_problem_approximate}) is smaller than that of (\ref{eq:subproblem_BPX}). Hence, the opposite optimal value of (\ref{eq:UAV_location_equal_problem_approximate}a) is the upper bound of that of (\ref{eq:subproblem_BPX}).

\subsubsection{UAV transmit power optimization}
Given the content delivery decision variable set ${\mathcal S}(t)$ and UAV trajectories ${\mathcal X}(t)$, it is still hard to mitigate (\ref{eq:subproblem_BPX}) owing to the non-convex objective function. Likewise, the SCA technique is explored to tackle the non-convexity. The following Proposition shows a method of optimizing UAV transmit power.
\begin{prop}\label{lemma:lemma_UAV_power}
{\rm By exploring an SCA technique, the UAV transmit power at time slot $t$ can be configured by mitigating the following convex optimization problem.
\begin{subequations}\label{eq:UAV_power_problem_approx}
\begin{alignat}{2}
& \mathop {\rm Maximize }\limits_{{\mathcal {P}}(t),{\{ \eta_{i}(t) \}}} {\mkern 1mu} \text{ } - V\rho \sum\nolimits_{j } {{p_j}(t)}  - \sum\nolimits_{j } {{{[{H_j}(t)]}^ + }{p_j}(t)}  + \nonumber \\
& \qquad \qquad {\sum\nolimits_{i} {\{ {{[{Q_{i}}(t)]}^ + } + {{[{Z_{i}}(t)]}^ + }\} {\eta _{i}}(t)}} \\
& {\rm s.t.} \text{ } \sum\nolimits_{j } {\left( {{s_{ij}}(t){{\hat R}_{ij}}(t) - {s_{ij}}(t)F_{ij}^{(r)}(t)} \right)}  - \sum\nolimits_{j } \left ( {s_{ij}}(t) \times \right. \nonumber \\
& \sum\nolimits_{k \in {\mathcal J}\backslash \{ j\} } {G_{ik,}^{(r)}(t)( {{p_k}(t) - p_k^{(r)}(t)} )}  ) \ge {\eta _{i}}(t),\forall i,t \\
& \quad \eta_i(t) \ge s_{ij}(t){{C_{i,f_i}^{\rm th}}(t )}, \forall i,j,t \\
& \quad {\rm (\ref{eq:original_problem}d), (\ref{eq:original_problem}g)}
\end{alignat}
\end{subequations}
where ${\hat R_{ij}}(t) = {\log _2}( {{\sigma ^2}W + \sum\nolimits_{k \in {\mathcal J}} {\frac{{{{p_k}(t)\theta _{ij}}}}{{{g^2} + ||{{\bm x}_k}(t) - {{\bm x}_{i}^{\rm u}(t)}|{|^2}}}} } )$, $F_{ij}^{(r)}(t) = {\log _2}( {{\sigma ^2}W + \sum\nolimits_{k \in {\mathcal J}\backslash \{ j\} } {p_k^{(r)}(t){h_{ik}}(t)} } )$, $G_{ik}^{(r)}(t) = \frac{{{h_{ik}}(t)}}{2^{F_{ij}^{(r)}(t)}\ln 2}$, and $p_k^{(r)}(t)$ is the given transmit power of UAV $k$ at the $r$-th iteration of SCA technique. (\ref{eq:UAV_power_problem_approx}c) is enforced due to the data rate requirement of enabling a user's desired QoE state.}
\end{prop}
\begin{proof}
Please refer to Appendix D.
\end{proof}

{{Remark 4}}: 
The objective function (\ref{eq:UAV_power_problem_approx}a) is linear. As the LHS of (\ref{eq:UAV_power_problem_approx}b) is concave w.r.t. $p_k(t)$, it is a convex constraint. Then, we can conclude that (\ref{eq:UAV_power_problem_approx}) is convex that can be efficiently alleviated by MOSEK. Similarly, the utilization of the approximation results in that the feasible set of (\ref{eq:UAV_power_problem_approx}) is a subset of that of (\ref{eq:subproblem_BPX}). Therefore, the optimal opposite value of (\ref{eq:UAV_power_problem_approx}a) is the upper bound of that of (\ref{eq:subproblem_BPX}). 

Based on the above derivation, the main steps of the iterative optimization scheme of solving (\ref{eq:subproblem_BPX}) can be summarized in Algorithm \ref{alg:alg1}.
\begin{algorithm}
\caption{Iterative UAV content delivery, trajectory, and transmit power optimization}
\label{alg:alg1}
\begin{algorithmic}[1]
\STATE \textbf{Initialization:} Randomly initialize ${\mathcal X}^{(0)}(t)$ and ${\mathcal P}^{(0)}(t)$, let $r = 0$.
\REPEAT
\STATE Given ${{\mathcal X}^{(r)}(t), {\mathcal P}^{(r)}(t)}$, solve (\ref{eq:UE_association_problem}) to obtain the optimal solution ${{\mathcal S}^{(r+1)}(t)}$.
\STATE Given ${{\mathcal S}^{(r+1)}(t), {\mathcal X}^{(r)}(t), {\mathcal P}^{(r)}(t)}$, solve (\ref{eq:UAV_location_equal_problem_approximate}) to generate the optimal solution ${{\mathcal X}^{(r+1)}(t)}$.
\STATE Given ${{\mathcal S}^{(r+1)}(t), {\mathcal X}^{(r+1)}(t), {\mathcal P}^{(r)}(t)}$, solve (\ref{eq:UAV_power_problem_approx}) to obtain the optimal solution ${{\mathcal P}^{(r+1)}(t)}$.
\STATE Update $r = r + 1$.
\UNTIL {Convergence or $r = r_{\max}$.}
\end{algorithmic}
\end{algorithm}

Finally, we can summarize the main steps of Lyapunov-based optimization framework of solving the original problem (\ref{eq:original_problem}) in Algorithm \ref{alg:alg2}.
\begin{algorithm}
\caption{Fresh, Fair, and Energy-Efficient Content Provision (F$^2$E$^2$CP)}
\label{alg:alg2}
\begin{algorithmic}[1]
\STATE \textbf{Initialization:} Initialize $Q_{i}(1) \in [0, 1]$, $Z_{i}(1) \in [0, 1]$, $H_j(1) \in [0, 1]$ for all user $i \in {\mathcal I}$, UAV $j \in {\mathcal J}$.
\FOR {each time slot $t = 1, 2, \ldots, T$}
\STATE Observe the virtual queues $Q_{i}(t)$, $Z_{i}(t)$, and $H_j(t)$.
\STATE Compute $\gamma_{i}(t)$ using (\ref{eq:compute_gamma}) for all user $i$.
\STATE Compute $b_{j,f_i}(t)$ using (\ref{eq:greedy_bk}) for all UAV $j$.
\STATE Obtain the content delivery decision set ${\mathcal S}(t)$, UAV trajectories ${\mathcal X}(t)$, and UAV transmit power ${\mathcal P}(t)$ using Algorithm \ref{alg:alg1}.
\STATE Calculate $p_j^{tot}(t)$ for all UAV $j$ using (\ref{eq:power_total_at_t}).
\STATE Calculate $u_{i}^{\rm dl}(t)$ for all user $i$ using (\ref{eq:achievable_data_rate}).
\STATE Update $Q_{i}(t+1)$, $Z_{i}(t+1)$, and $H_j(t+1)$ for all user $i$ and UAV $j$ using (\ref{eq:Queue_ui}), (\ref{eq:Queue_Z}), and (\ref{eq:Queue_H}), respectively.
\ENDFOR
\end{algorithmic}
\end{algorithm}

\section{Performance analysis}
In this section, the convergence performance of the proposed algorithms is analyzed. Observing that the PAoI of a data packet varies with many time-varying factors (e.g., inter-arrival time of the packet and packet queueing delay), the analysis on the PAoI of a data packet in the average sense is conducted.

\subsection{Convergence analysis}
The following Lemma shows that the convergency and validity of Algorithms 1 and 2 can be guaranteed.
\begin{lemma}\label{lemma:lemma_convergent}
{\rm Algorithm \ref{alg:alg1} is convergent, and Algorithm 2 can make all virtual queues mean-rate stable.}
\end{lemma}
\begin{proof}
Please refer to Appendix E.
\end{proof}


\subsection{Analysis of expected PAoI}
Recall that many factors such as the inter-arrival time of a packet, packet queueing delay in the preprocessing queue, packet preprocessing time, and the edge arrival duration will contribute to the PAoI of a packet. We next analyze these factors in detail.

\subsubsection{Packet queueing delay analysis}
The analysis on the packet queueing delay requires the study of packet queueing behavior.
To study the packet queueing behavior, the queue evolution process that involves the arrival, accumulation and departure of packets should be investigated.

During time $q$, once received the content requests of users during time $q$, the content server will generate and send out packets required by the users.
To facilitate the analysis, a Poisson distribution with intensity (or the average number of ne{\underline{\textbf w}} {packets}) {$\vartheta_{w}(q)$} is explored to model the random and independent packet arrivals in the BS.
As mentioned in the system model, once arrival, new packets will not be processed immediately and will enter the preprocessing queue to wait to be preprocessed.
If the packet preprocessing rate, which is determined by the CPU computing speed \cite{xu2019optimizing}, is slower than the packet arrival rate, some packets will be accumulated in the preprocessing queue. After being preprocessed, packets will depart from the preprocessing queue. In this regard, the packet departure rate is equal to the packet preprocessing rate.
Denote by $N_{a}(q)$ the {\underline {\textbf a}}ccumulated number of packets in the preprocessing queue at time $q$.
The value of $N_{a}(q)$ is simultaneously determined by the following three factors: a) the accumulated number of packets; b) the number of new arrivals during $q - 1$, which will be counted at at time $q$; c) the packet departure rate;
thus, we can present the evolution model of $N_{a}(q)$ as follows
\begin{equation}\label{eq:queue_evolution}
{N_{a}^q} = \left\{ {\begin{array}{*{20}{l}}
{0, q = 1}\\
{{{[{N_{w}^{q-1}} - {n_c}]}^ + }, q = 2}\\
{[{N_{a}^{q-1}} + {N_{w}^{q-1}} - }
{{n_c}{]^ + }, q \ge 3}
\end{array}} \right.
\end{equation}
where we write $x^q$ instead of $x(q)$ to lighten the notation. 
$N_{w}^{q-1}$ is the number of ne{\underline{\textbf w}} arrivals in time $q$-$1$,
the subtraction operation of ${n_c}$ is performed because ${n_c} = \frac{r}{\kappa l}$ packets can be preprocessed during each time interval, $l$ is the packet size (in bits), $r$ is the CPU computing speed of the preprocessor with the units CPU cycles per second, and $\kappa$ is a scaling parameter depending on the specific operation conducted on the packet with the units being CPU cycles per bit \cite{xu2019optimizing}.
As new arrivals in the $1$-st time interval will be counted in the $2$-nd time interval and there are no accumulated packets in the $1$-st time interval, we have $N_a^1 = 0$.

Based on the above evolution model, we have the following Lemma that derives the closed-form expressions of the accumulated number of packets in the preprocessed queue and the packet queueing delay of a newly arrival packet at time $q$.
\begin{lemma}\label{lem:lemma_accumulated}
{\rm The accumulated number of packets in the queue can be approximated as Poisson distribution. As such, the average number of accumulated packets in the queue at time $q > 1$ can be derived as
\begin{equation}\label{eq:accu_packets}
\vartheta_{a}^q = \left[\vartheta_w^{q-1} + \vartheta_a^{q-1} - n_c (1 - e^{-\vartheta_w^{q-1}-\vartheta_a^{q-1}})\right]^+.
\end{equation}
Besides, the average packet queueing delay of a newly arrival packet $m$ sent to user $i$ at time $q > 1$ is $\vartheta_{a}^q Y_{i,m}^{\rm S}$.
}
\end{lemma}
\begin{proof}
Please refer to Appendix F.
\end{proof}

\subsubsection{Analysis of the expected edge arrival duration}
We know that a content file including a batch of packets will be delivered from the BS or a UAV to a destination user in each time slot $t$.
Then, we attempt to tackle the following issue: \emph{What is the expected edge arrival duration for transmitting a packet?} The following Proposition gives the answer.
\begin{prop}\label{prop:prop_low_latency}
{\rm For any user $i$ and packet $m$, the expected edge arrival duration of packet $m$ destined to user $i$ can be given by
\begin{equation}\label{eq:exp_edge_delay}
{\mathbb E}[Y_{i,m}^{\rm A}] = \frac{(l+L)N\Delta_ t}{2JL}
\end{equation}
where ${\mathbb E}[\cdot]$ represents an expectation operation.
}
\end{prop}
\begin{proof}
Please refer to Appendix G.
\end{proof}

\subsubsection{Closed-form expression of expected PAoI}
Based on the obtained results in Lemma \ref{lem:lemma_accumulated} and Proposition \ref{prop:prop_low_latency}, we can derive the closed-form expression of the expected PAoI of a packet in the following Lemma.
\begin{lemma}\label{lemma:EPAoI}
{\rm The expected PAoI of packet $m$, which is generated at time $q$ and sent to user $i$, $\forall i$, can be expressed as
\begin{equation}\label{eq:average_PAoI}
{\mathbb E}[{\Delta _{i,m}}({{\hat q}; q})] = \left\{ {\begin{array}{*{20}{l}}
{ \frac{1}{n_c} + \frac{(l+L)N\Delta_ t}{2JL}}&{{q} = 1,m = 1}\\
\frac{1}{{{\lambda _i^q}}} + \frac{1}{n_c} +  \frac{\vartheta_a^q}{n_c} + \frac{(l+L)N\Delta_ t}{2JL} &{{q} > 1,m > 1}
\end{array}} \right.
\end{equation}
where $\lambda_i^q = \vartheta_w^q {\rm Pr}_i$ is the intensity of new arrival packets destined to user $i$ at time $q$.
}
\end{lemma}
\begin{proof}
Please refer to Appendix H.
\end{proof}

\section{Simulation Results}

\subsection{Comparison algorithms and parameter setting}
To verify the effectiveness of the proposed Algorithm \ref{alg:alg2}, we compare it with five benchmark algorithms:
1) \emph{Static UAV caching with power control (SUWPC) algorithm:} It randomly generates horizontal locations for $N$ hovering UAVs with the same deployment altitude $g$. The proposed UAV transmit power optimization method is adopted in the algorithm. Besides, each UAV randomly delivers a cached or forwarded content file to terrestrial users within its LoS coverage region at each time slot after receiving users' content requests;
2) \emph{Static UAV caching without power control (SUPC) algorithm:} The unique difference between SUPC and SUWPC lies in that SUPC adopts the maximum UAV transmit power to deliver content files to terrestrial users;
3) \emph{Circular trajectory-based joint optimization (CTJO) algorithm:} Each UAV flies in a circular trajectory with a speed of 10 m/s. At the beginning of the simulation, UAVs are deployed in line with an equal interval. The distance between two adjacent UAVs is $1/4J$ km. The horizontal locations of the first and the last UAVs are $(1/4+1/\left\lfloor 8J \right\rfloor,1/4)$ km and $(1/2-1/\left\lfloor 8J \right\rfloor,1/4)$ km, respectively, and turning radiuses of them are $1/\left\lfloor {8N} \right\rfloor$ km and $1/4-1/\left\lfloor 8J \right\rfloor $ km. In CTJO, all UAVs are deployed at the same altitude $g$.
Besides, the proposed UAV caching and UAV transmit power optimization methods are utilized in this algorithm.
4) \emph{Circular trajectory-based UAV caching (CTUC) algorithm:} The unique difference between CTUC and CTJO is that CTUC adopts the maximum UAV transmit power to deliver content files to terrestrial users.
5) \emph{Circular trajectory without UAV caching (CTWUC) algorithm:} The unique difference between CTWUC and CTJO is that UAVs do not cache content files and will randomly forward content files from the BS to users in CTWUC.

We consider an urban area of size $0.5 \times 0.5$ km$^2$ with highrise buildings, where terrestrial users are randomly distributed and moving. This scenario corresponds to the AtG channel environment of the worst case\cite{Al2014Optimal}. The radio frequency propagation parameters are: $\eta_{LoS} = 2.3$, carrier frequency $f_c = 4.9$ GHz, speed of light $c = 3.0 \times 10^8$ m/s, $\theta^{\rm th} = 70^{\circ}$, $\sigma^2 = -174$ dBm/Hz, total bandwidth $W = 100$ MHz, $ \hat D = 24$ s, $D^{\rm ul} = 5$ s, $L = 150$ Mbits, $D^{\rm th} = 0.6$, $\Delta_t = \hat D - D^{\rm th}(\hat D - L/(Wu_{\rm dl}^{\max}))$ s \cite{Al2014Optimal,Chen2017Caching}.
The default values of parameters related to UAVs and terrestrial users are: ${\tilde p_j} =450$ mW, $ \hat p_j = 500$ mW, $p_j^c = 20$ mW, $p_j^{\min}= 1$ mW, $\forall j$, $p_{\max} = 480$ mW, $e_{\max} = 250$ m, $g = 200$ m, $d_{\min} = 50$ m, $J = 4$, and $N = 50$.
Expected PAoI analysis related parameters are: $ l = 5000$ Bytes, $\gamma/\kappa  = 1$ Mbits/s, $\vartheta_w^q \sim {\rm Pois}(\frac{\gamma}{\tau l })$, and ${\rm Pr}_i = 1/N$, $\forall i, q$ \cite{xu2019optimizing}.
Other system parameters are listed as follow: $r_{\max} = 200$, $V = 0.01$, $T = 500$, and $\rho = 0.1$ \cite{yang2020repeatedly}.

\subsection{Performance evaluation}
In this subsection, we conduct simulations to comprehensively understand the availability of the developed F$^2$E$^2$CP algorithm.
Besides, we repeat all comparison algorithms for fifteen times to eliminate the influence of some randomly initialized parameters such as UAV transmit power and locations on the performance evaluation. The final simulation result is the average of the obtained fifteen results.

First, we design a simulation to verify the effectiveness of the F$^2$E$^2$CP algorithm under the default parameter setting.
Particularly, we plot the tendency of the virtual queue stability values, defined as ${S_Q} = \mathop {\max }\nolimits_i \  {[{Q_i}(t)]^ + }/t$, ${S_Z} = \mathop {\max }\nolimits_i \  {[{Z_i}(t)]^ + }/t$, and ${S_H} = \mathop {\max }\nolimits_{j} \  {[{H_j}(t)]^ + }/t$ in Fig. \ref{fig_queue_stability}. Besides, the two-dimensional (2D) trajectories of four UAVs in the first 150 time slots and their final 2D positions are plotted in Fig. \ref{fig_uav_track}.
\begin{figure}[!t]
\centering
\begin{minipage}[t]{0.24\textwidth}
\centering
\includegraphics[width=1.7 in]{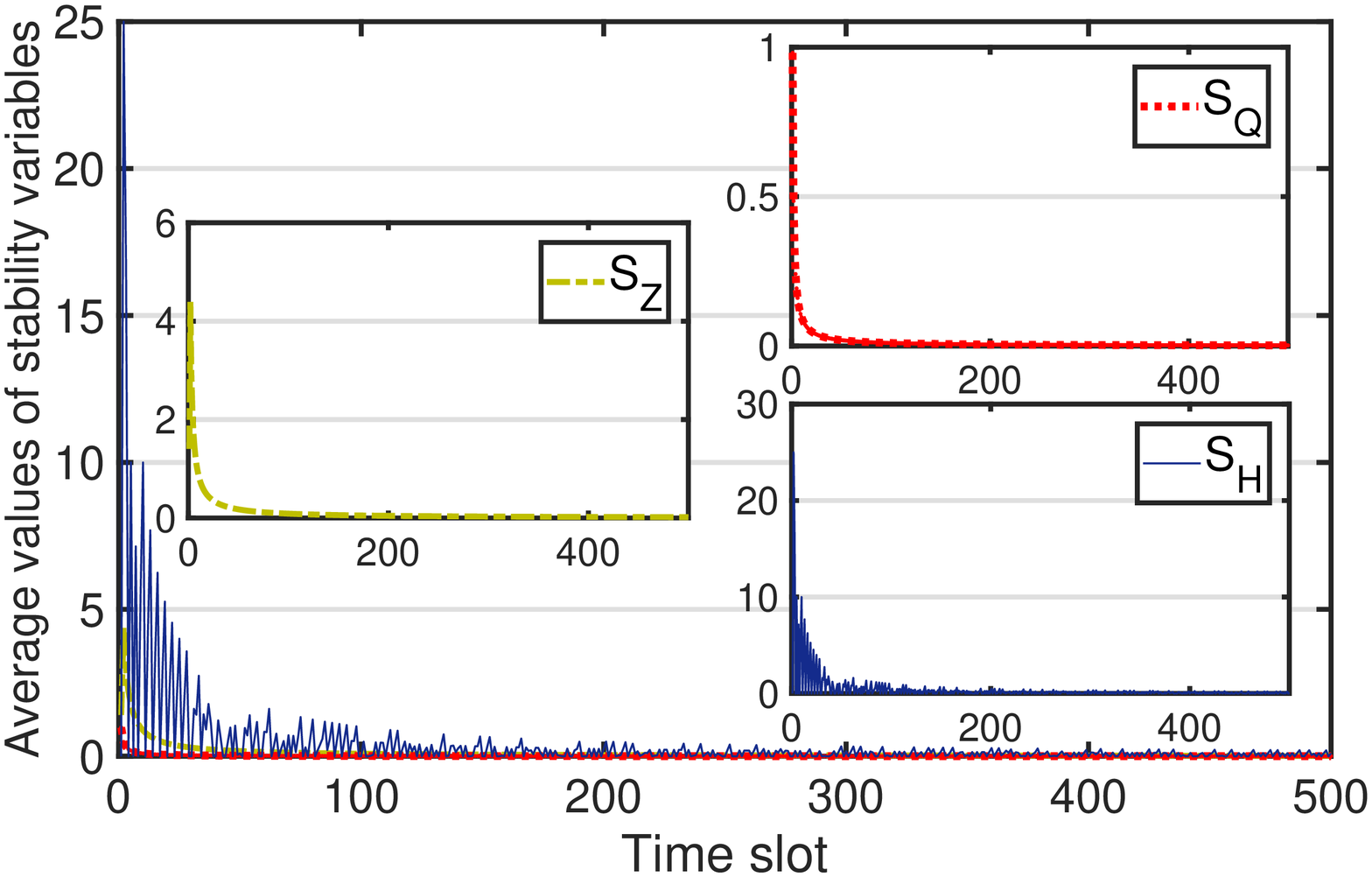}
\caption{Trend of virtual queue stability values vs. time slot.}
\label{fig_queue_stability}
\end{minipage}
\begin{minipage}[t]{0.24\textwidth}
\centering
\includegraphics[width=1.7 in]{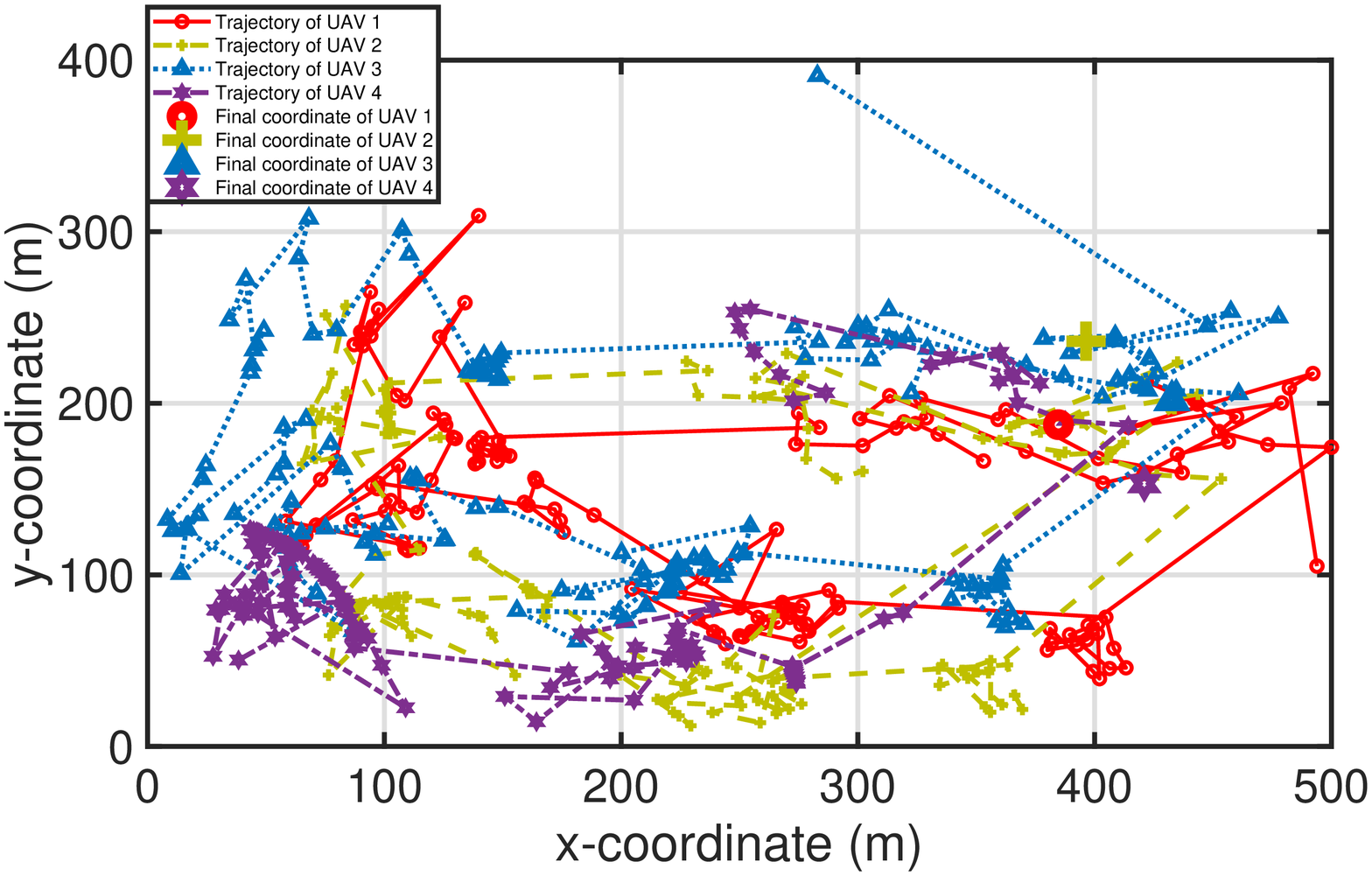}
\caption{Trajectories of four UAVs projected in a 2D space.}
\label{fig_uav_track}
\end{minipage}
\end{figure}

The following observations can be achieved from these figures: 1) the obtained queue stability values are bounded during the whole content provision period;
2) the obtained queue stability values tend to zero with an increasing time slot; as a result, all introduced virtual queues are mean-rate stable according to (\ref{eq:Queue_EQ}), i.e., all time average constraints in (\ref{eq:Jensen_problem}) can be imposed. This result verifies the effectiveness of the proposed F$^2$E$^2$CP algorithm;
3) the movement constraints of each UAV can be satisfied at each time slot.

Next, we conduct a simulation to testify the performance of F$^2$E$^2$CP by comparing it with other five benchmark algorithms.
To quantify the algorithm performance, the following key performance indicators are introduced:
the network profit that is defined as $\phi ( \{{\bar u_{i}^{\rm dl}}(T)\} )$, the total UAV transmit power consumption that is calculated by $\sum\nolimits_{j } {\bar p_j^{tot}(T)}$,
the energy efficiency that is computed by (\ref{eq:original_problem}a), the Jain's fairness index, defined as ${{{\left( \sum\nolimits_{i}{{{{\bar{u}}}_{i}^{\rm dl}}} \right)}^{2}}}/{N\sum\nolimits_{i}{(\bar{u}_{i}^{\rm dl})^{2}}}$ with ${{\bar u}_{i}^{\rm dl}} = \frac{1}{T}\sum\nolimits_{t = 1}^T {{u_{i}^{\rm dl}}(t)} $, the expected PAoI of a packet computed by (\ref{eq:average_PAoI}), where the obtained ${\mathbb E}[Y_{i,m}^{\rm A}] = {N(l+L)T\Delta_t}/(2L{\sum\nolimits_{t = 1}^T {\sum\nolimits_{i} {\sum\nolimits_{j} {s_{ij}(t)}} }})$.

\begin{figure}[!t]
\centering
\subfigure[Network profit vs. the number of users]{\includegraphics[width=1.7 in]{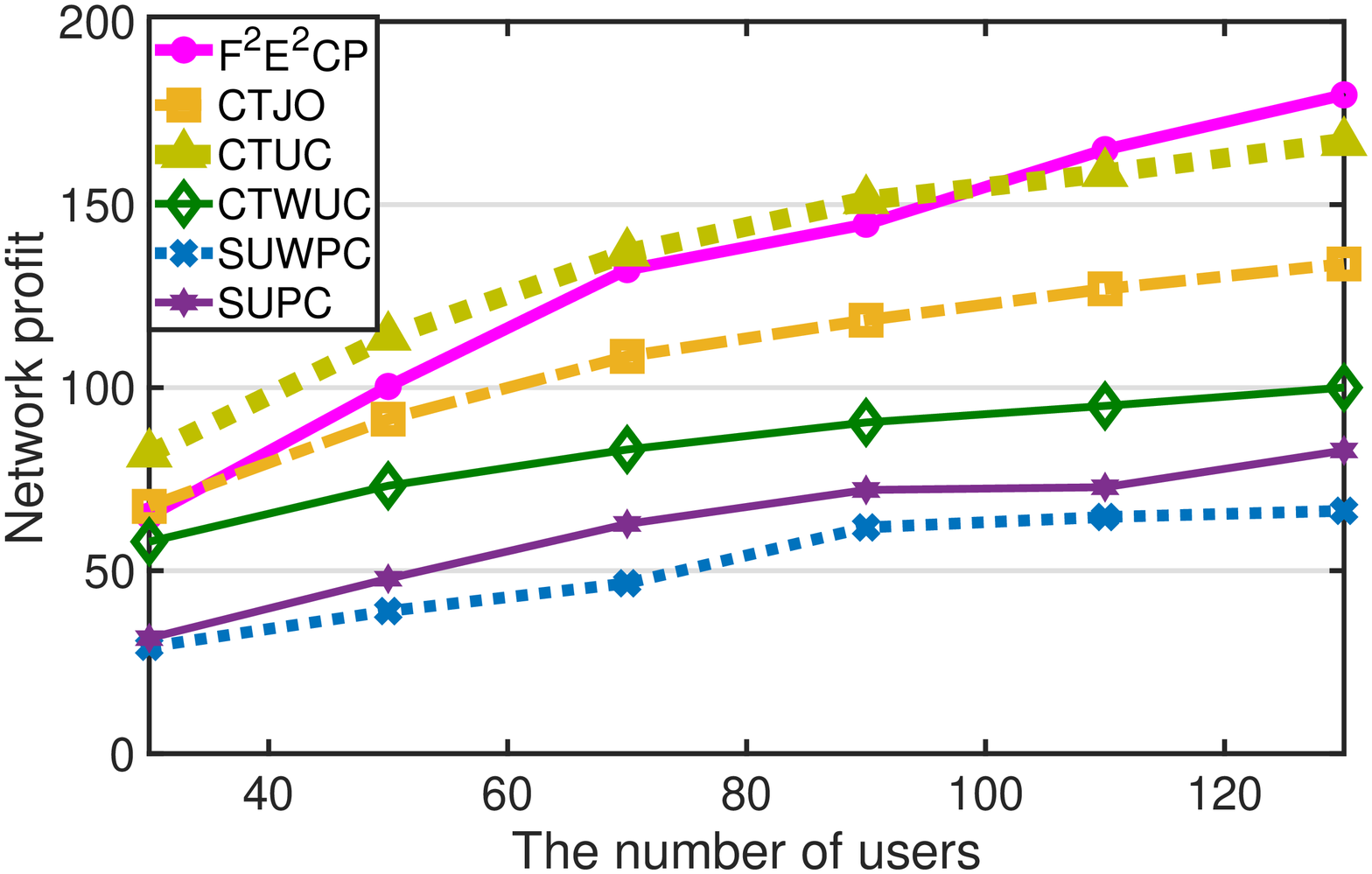}%
\label{fig_profit_user}}
\subfigure[Network profit vs. the number of UAVs]{\includegraphics[width=1.7 in]{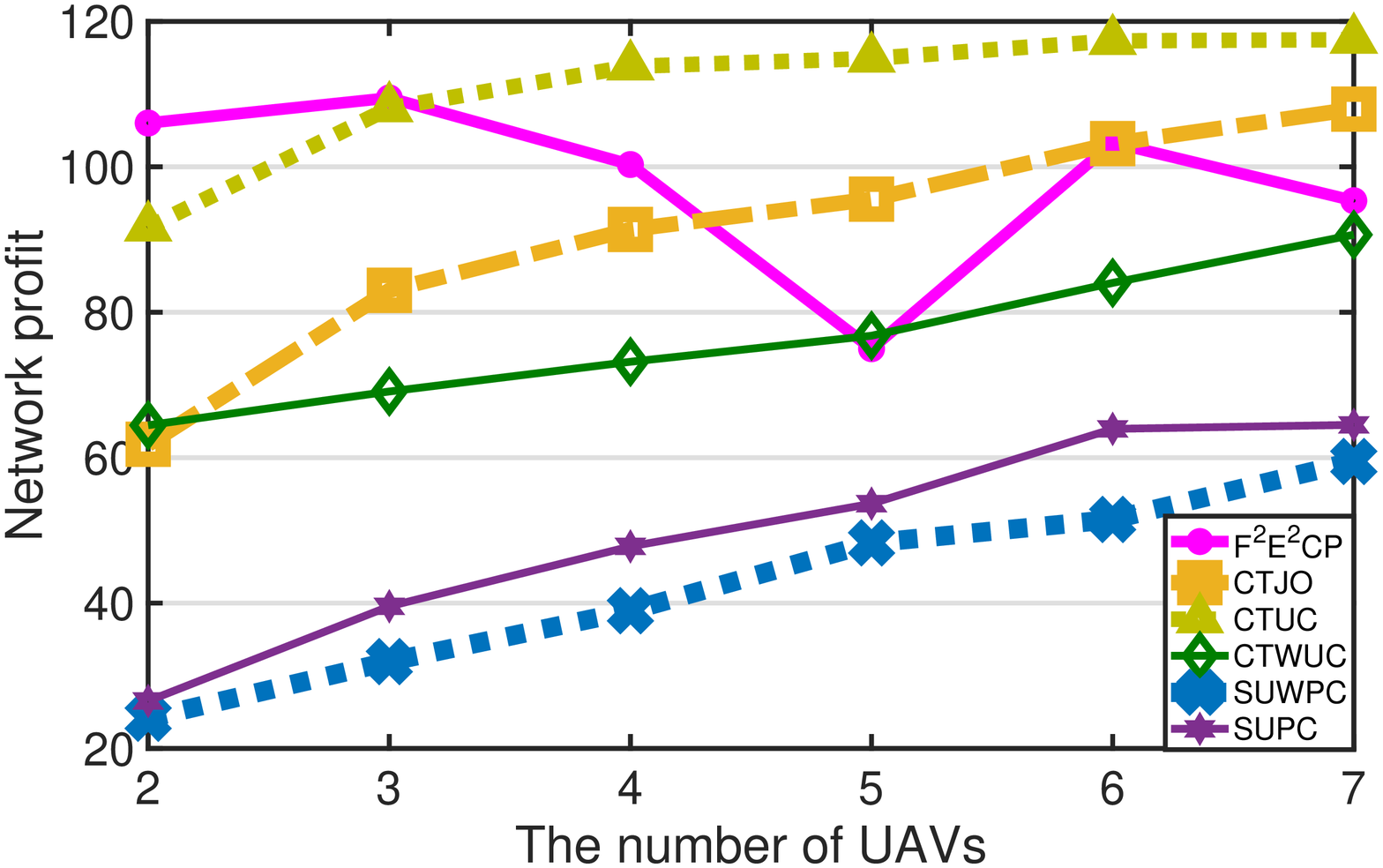}%
\label{fig_profit_uav}}
\caption{Comparison of the obtained network profit.}
\label{fig_profit}
\end{figure}


\begin{figure}[!t]
\centering
\subfigure[Total UAV power consumption vs. $N$]{\includegraphics[width=1.7 in]{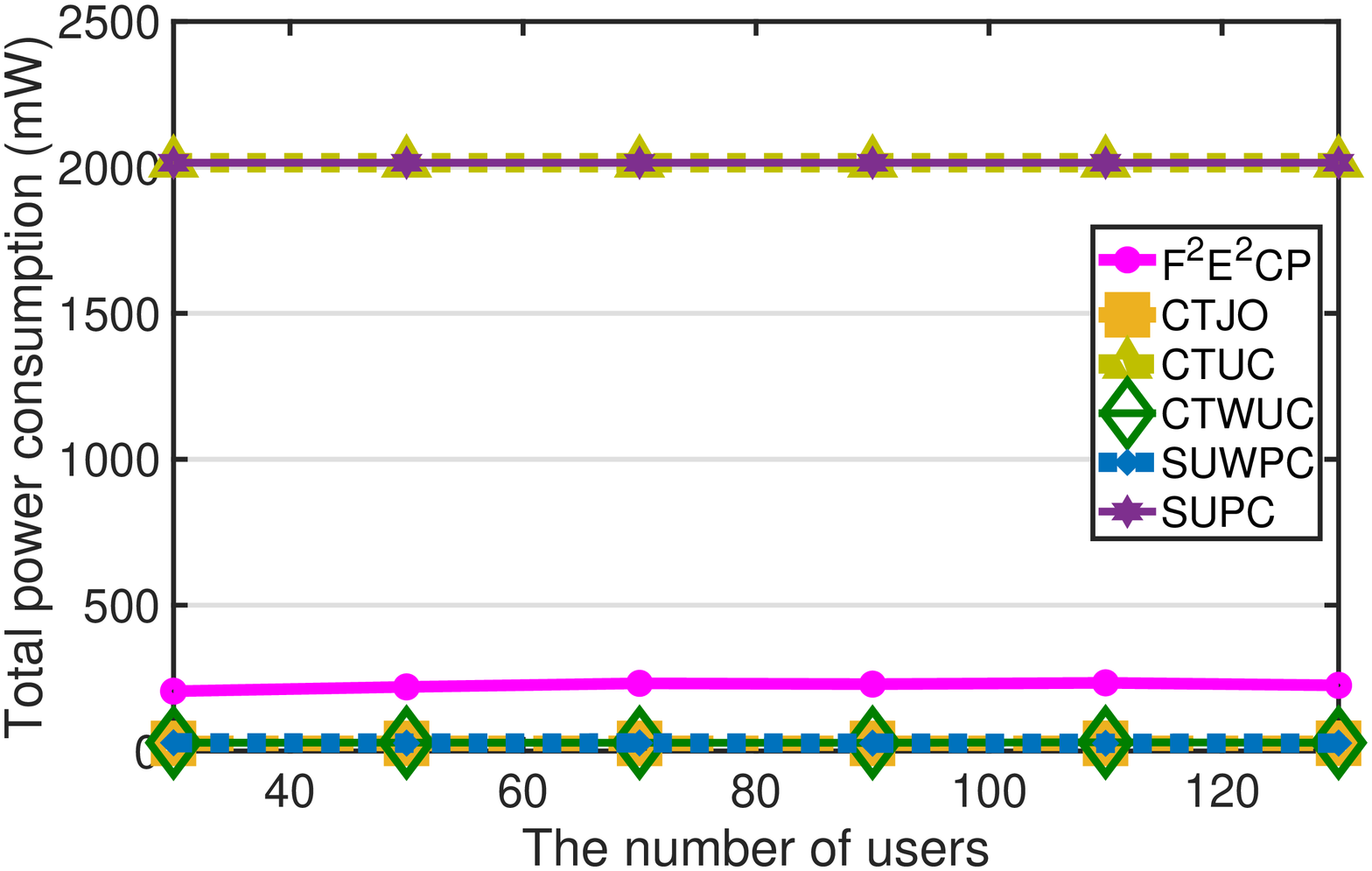}%
\label{fig_power_user}}
\subfigure[Total UAV power consumption vs. $J$]{\includegraphics[width=1.7 in]{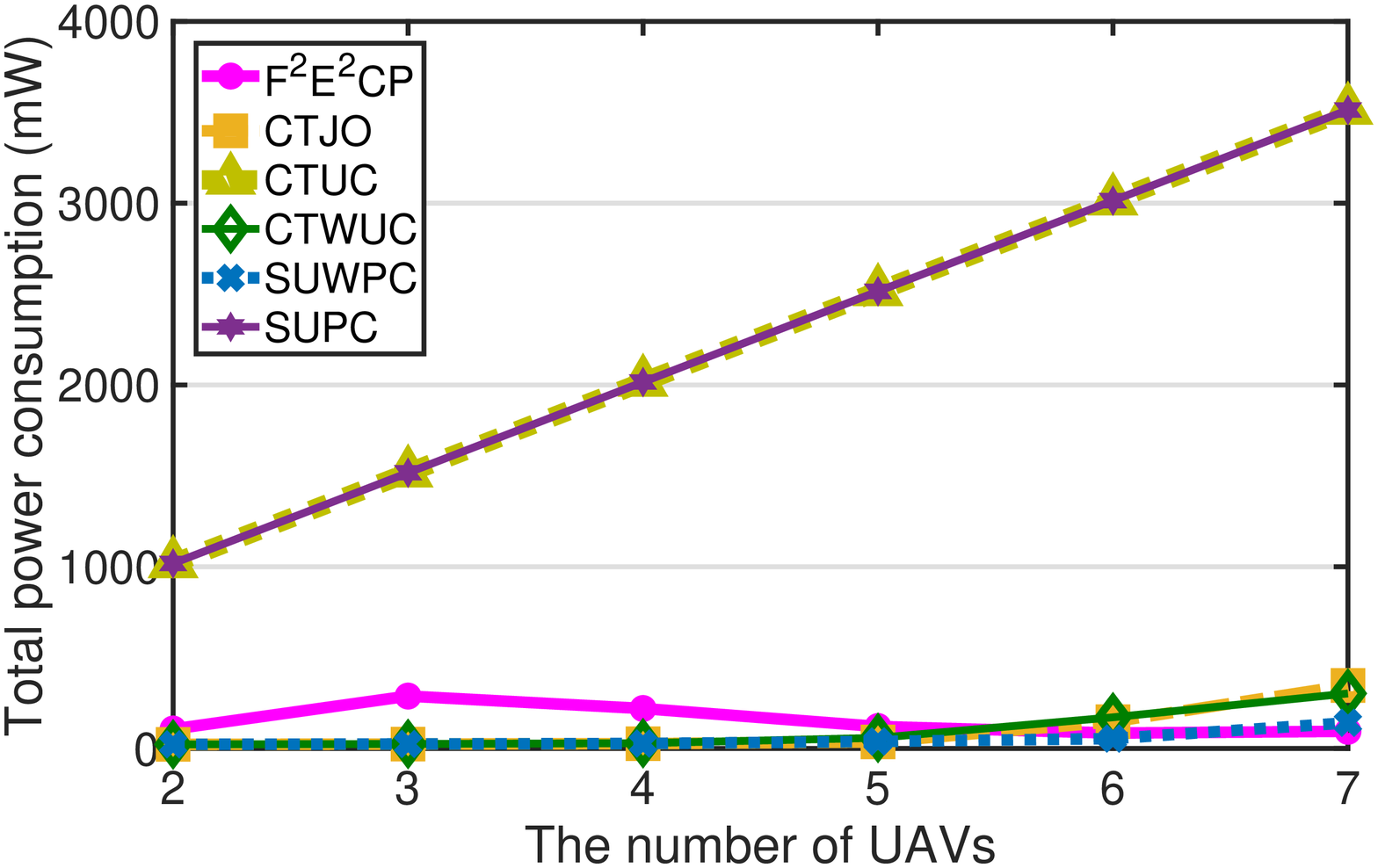}%
\label{fig_power_uav}}
\caption{Comparison of the obtained total UAV power consumption of all comparison algorithms.}
\label{fig_power}
\end{figure}

In Figs. \ref{fig_profit} and \ref{fig_power}, the influence of the number of users and UAVs on the obtained network profit and the total UAV power consumption is plotted. The tendency of the obtained energy efficiency and Jain's fairness indexes of all comparison algorithms are depicted in Figs. \ref{fig_EE_user_num} and \ref{fig_fairness_user_num}, respectively. Besides, in Fig. \ref{fig_EE_user_num}, the energy efficiency results of two comparison algorithms (i.e., CTUC and SUPC) are not plotted. This is because they suggest to transmit content files with the maximum UAV transmit power, and will achieve smaller energy efficiency values than other UAV power control-based algorithms. In Fig. \ref{fig_fairness_user_num}, the fairness index results of CTJO and SUWPC are not plotted. As CTJO and SUWPC achieve smaller network profit than CTUC and SUPC, CTJO and SUWPC will obtain smaller Jain's fairness indexes.
\begin{figure}[!t]
\centering
\begin{minipage}[t]{0.24\textwidth}
\centering
\includegraphics[width=1.7 in]{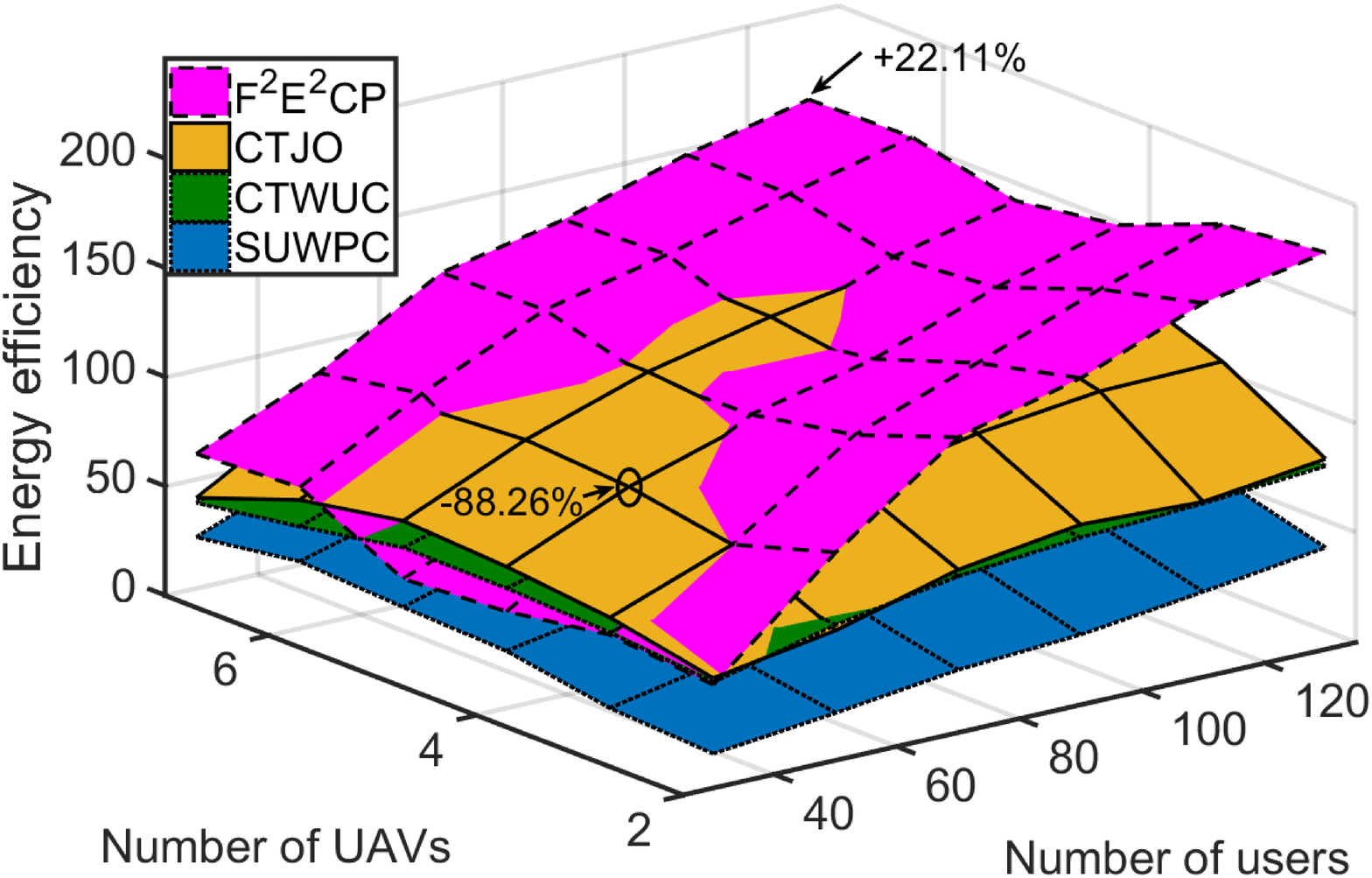}
\caption{Energy efficiency vs. $N$ and $J$.}
\label{fig_EE_user_num}
\end{minipage}
\begin{minipage}[t]{0.24\textwidth}
\centering
\includegraphics[width=1.7 in]{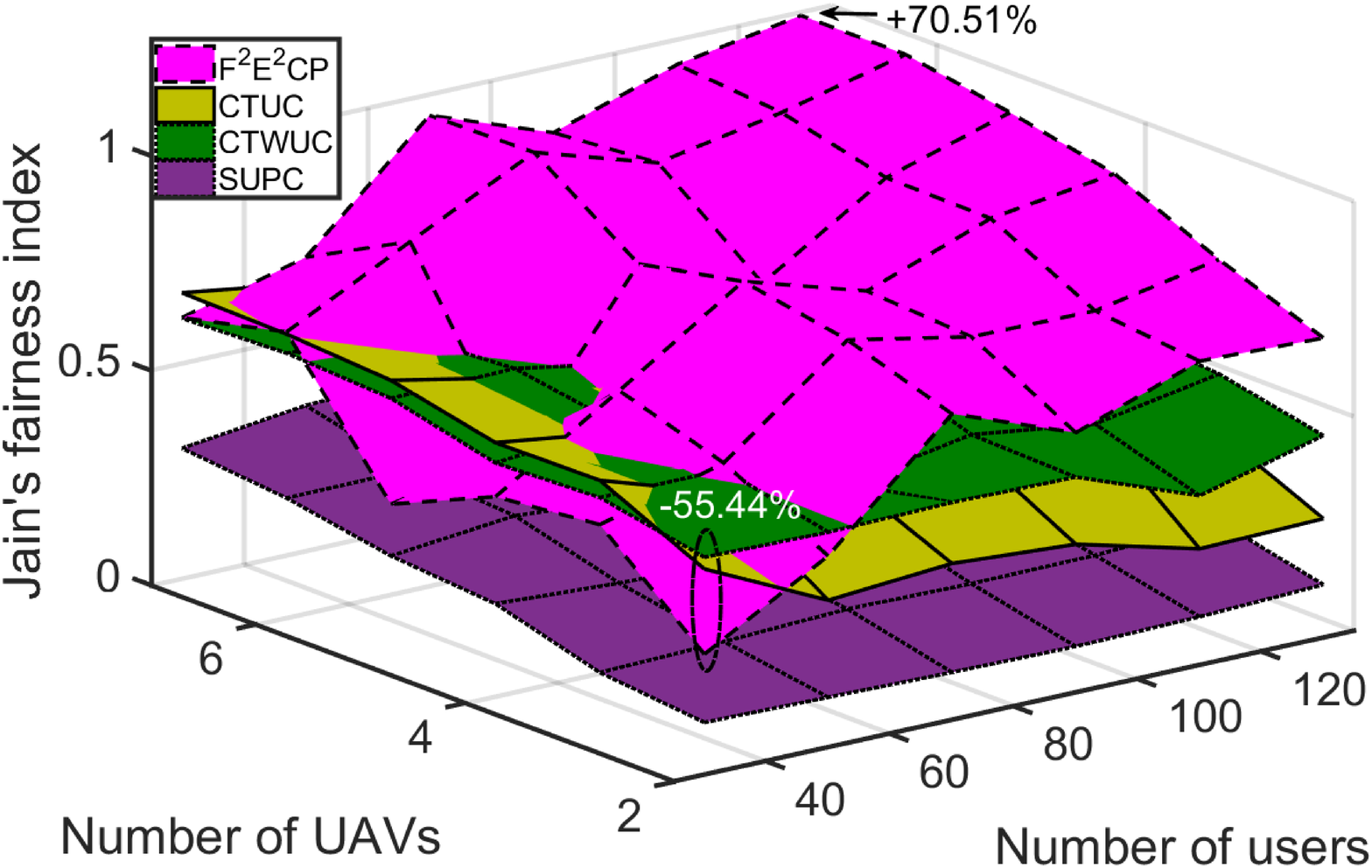}
\caption{Jain's fairness index vs. $N$ and $J$.}
\label{fig_fairness_user_num}
\end{minipage}
\end{figure}

From these figures, we can observe that:
1) for all algorithms, their achieved network profit increases with an increasing number of users because content files can be successfully delivered to more users;
2) CTUC obtains the greatest network profit when $N \le 90$, followed by the proposed F$^2$E$^2$CP algorithm. This is mainly because CTUC adopts the maximum UAV transmit power that results in great AtG data rates.
Although SUPC adopts the scheme of transmitting content files with the maximum UAV transmit power, small network profit is obtained as many terrestrial users are experiencing outage throughout the content provision period.
The observations show that the design of suitable UAV trajectories can help improve network profit; 
3) except for F$^2$E$^2$CP, the obtained network profit of other comparison algorithms increase with an increasing number of UAVs. For F$^2$E$^2$CP, a great number of UAVs does not lead to big network profit. This is due to the complex impact of signal interference. For other comparison algorithms, the distance between any two UAVs is large. As a result, they can increase UAV transmit power to achieve greater network profit. Yet, the nearest distance between two UAVs can be $50$ m which may result in large signal interference. To reduce signal interference, F$^2$E$^2$CP decreases the UAV transmit power when $J \ge 4$;
4) the generated total power consumption of all algorithms almost does not vary with the number of users because a UAV can deliver a content file to one user at each time slot.
When $J < 5$, F$^2$E$^2$CP consumes greater UAV transmit power than other three comparison algorithms adopting the proposed UAV transmit power optimization method such that greater network profit can be gained;
5) when $N \le 90$ and $J \le 5$, CTJO may be energy-efficient than F$^2$E$^2$CP. For example, when $N= 50$ and $J = 4$, the obtained energy efficiency of the proposed F$^2$E$^2$CP algorithm is $88.26$\% of CTJO.
However, from Fig. \ref{fig_profit}, we can find that F$^2$E$^2$CP may obtain higher network profit than CTJO under the similar parameter setting.
This observation shows that F$^2$E$^2$CP tends to deliver content files to fewer users continuously when users are sparsely distributed, resulting from a small $N$.
When the number of users and UAVs becomes great, F$^2$E$^2$CP can achieve greater energy efficiency than the CTJO algorithm. For instance, when $N= 130$ and $J = 7$, F$^2$E$^2$CP is energy-efficient by $22.11$\% than CTJO.
The above result indicates that a simple circular UAV trajectory is preferable when terrestrial users are sparsely distributed in the considered communication area. However, when users are densely distributed, optimizing UAV trajectory will greatly improve the energy efficiency of content provision. 
Besides, CTJO may be energy-efficient than CTWUC, which means that the UAV caching can improve the energy efficiency of content provision. SUWPC provides inefficient content delivery due to the static UAV deployment;
6) similarly, when the number of users is small (i.e., $N \le 50$), F$^2$E$^2$CP may be overwhelmed by CTUC and CTWUC in terms of fair content provision. For example, when $N = 30$ and $J = 2$, the obtained fairness index of F$^2$E$^2$CP is $58.62$\% and $55.44$\% of that of CTUC and CTWUC, respectively. However, when $N > 50$, F$^2$E$^2$CP can provide fairer content delivery services for terrestrial users than CTUC, CTWUC, and SUPC. For instance, the obtained Jain's fairness index of F$^2$E$^2$CP is $1.70$, $1.88$, and $3.04$ times of CTWUC, CTUC, and SUPC, respectively, when $N = 130$ and $J = 7$.

We plot the tendency of expected PAoI obtained by all comparison algorithms in Fig. \ref{fig_expected_PAoI} to show whether the expected PAoI can be reduced by the joint design of UAV caching, UAV trajectory, and UAV transmit power. 
In Fig. \ref{fig_expected_PAoI}, the expected PAoI is obtained with $N = 30$ and $J = 5$.

From this figure, we can observe: 1) the proposed F$^2$E$^2$CP algorithm achieves a small value of expected PAoI that is close to the theoretical value. The obtained values of ${\mathbb E}[\Delta_{i,m}(\hat q; q)]$ of other comparison algorithms are much greater than the theoretical value. Besides, the two circular UAV trajectory-based algorithms with UAV caching achieve smaller ${\mathbb E}[\Delta_{i,m}(\hat q; q)]$ than the two static UAV trajectory-based algorithms. This observation shows the importance of UAV trajectory optimization in delivering fresh packets to terrestrial users;
2) CTWUC obtains the largest ${\mathbb E}[\Delta_{i,m}(\hat q; q)]$ and takes at least $12.57$ seconds more than CTJO and CTUC to deliver packets to users. It indicates that UAV caching can significantly reduce the latency of delivering packets to users;
3) it is interesting to find that the UAV power control method alone cannot effectively reduce the expected edge arrival duration. For example, CTUC achieves smaller ${\mathbb E}[\Delta_{i,m}(\hat q; q)]$ than CTJO, yet, SUPC obtains greater ${\mathbb E}[\Delta_{i,m}(\hat q; q)]$ than SUWPC.

\begin{figure}[!t]
\centering
\includegraphics[width=2.0 in]{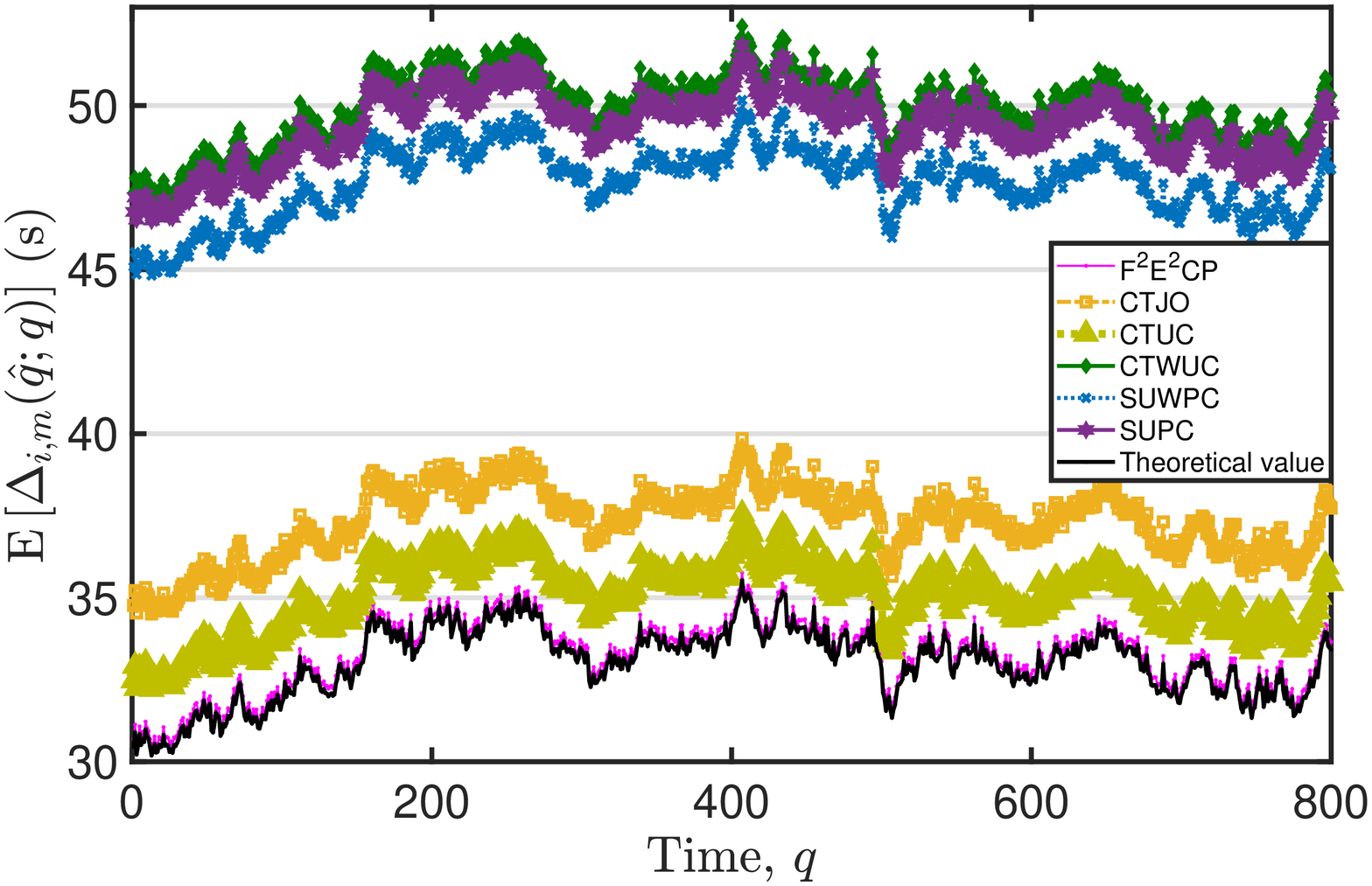}%
\caption{Comparison of expected PAoI of all comparison algorithms.}
\label{fig_expected_PAoI}
\end{figure}

Summarily, the above simulation results indicate that the joint design of UAV caching, UAV trajectory, and UAV transmit power can guarantee the provision of fresh, fair, and energy-efficient content files for terrestrial users.

\section{Conclusion and future work}
This paper investigated a private and cache-enabled UAV network for providing fresh, fair, and energy-efficient content delivery services for terrestrial users. To achieve this goal, we formulated a joint UAV caching, UAV trajectory, and UAV transmit power optimization problem.
We proposed a novel algorithm by leveraging a Lyapunov-based optimization framework integrating an iterative optimization scheme and an SCA technique to solve the formulated problem.
Besides, we discussed the convergence behaviour of the proposed algorithm as well as derived the theoretical value of expected PAoI of data packets.
Simulation results verified that the proposed algorithm could provide fresh content files for users and was more $22.11$\% and $70.51$\% energy-efficient and fairer than benchmark algorithms. This paper explored an optimization method to design UAV trajectories in a 2D space. How to incorporate reinforcement learning with optimization to solve the joint UAV caching, three-dimensional UAV trajectory, and UAV transmit power optimization problem is a topic worthy of research in the near future.

\appendix
\subsection{Proof of Proposition \ref{prop:time_decouple}}
Suppose all limits in (\ref{eq:Jensen_problem}) exist, the constraint (\ref{eq:Jensen_problem}b) is therefore equivalent to ${{\bar u}_{i}^{\rm dl}}(t) = {{\bar \gamma }_{i}}(t)$. $\bar g(t) \le \phi ({{\bar u}_{1}^{\rm dl}}(t), \ldots ,{{\bar u}_{N}^{\rm dl}}(t))$ can then be achieved. It means that the maximum value of the objective function of (\ref{eq:Jensen_problem}) is no greater than that of (\ref{eq:original_problem}). Besides, the maximum value of the objective function of (\ref{eq:original_problem}) can be obtained through letting ${{\bar \gamma }_{i}}(t) = {\bar u}_{i}^{\rm dl; \star}(t)$ for all $i \in {\mathcal I}$, and $t \in \{1, 2, \ldots\}$ with $({\bar u}_{1}^{\rm dl; \star}(t), \ldots, {\bar u}_{N}^{\rm dl; \star}(t))$ being the optimal time average achievable data rates of all users for (\ref{eq:original_problem}) \cite{Neely2014A}. It indicates that the feasible domain of (\ref{eq:original_problem}) is smaller than that of (\ref{eq:Jensen_problem}). Therefore, (\ref{eq:Jensen_problem}) and (\ref{eq:original_problem}) are equivalent. This completes the proof.

\subsection{Proof of Lemma \ref{lemma:1}}
We discuss the upper bound of $\frac{1}{2}{({[ {Q_{i}(t + 1)} ]^ + })^2}$ in three cases. According to {(\ref{eq:Queue_ui})} and the non-negative operation,

Case I: when ${Q_{i}(t + 1)}  \ge 0$ and ${Q_{i}(t)}  \ge 0$, we can obtain
\begin{equation}\label{eq:31}
\begin{array}{*{20}{l}}
{\frac{1}{2}{{({{[{Q_{i}}(t + 1)]}^ + })}^2}} = \frac{1}{2}{{({{[{Q_{i}}(t)]}^ + })}^2} + \\
\qquad \qquad \qquad {{[{Q_{i}}(t)]}^ + }({\varphi {C_{i,f_i}^{\rm th}}(t )} - {u_{i}^{\rm dl}}(t)) + \\
\qquad \qquad \qquad { \frac{1}{2}{{({\varphi {C_{i,f_i}^{\rm th}}(t )} - {u_{i}^{\rm dl}}(t))}^2}}
\end{array}
\end{equation}

Case II: when ${Q_{i}(t + 1)} \ge 0$ and ${Q_{i}(t)} < 0$, we can achieve ${\varphi {C_{i,f_i}^{\rm th}}(t )}-{{u}_{i}^{\rm dl}}(t)>{{Q}_{i}}(t+1)\ge 0$, ${{[ Q_{i}(t) ]}^{+}}=0$ and
\begin{equation}\label{eq:31}
\begin{array}{*{20}{l}}
{\frac{1}{2}{{({{[{Q_{i}}(t + 1)]}^ + })}^2} < \frac{1}{2}{{({\varphi {C_{i,f_i}^{\rm th}}(t )} - {u_{i}^{\rm dl}}(t))}^2}}\\
\qquad \qquad { = \frac{1}{2}{{({{[{Q_{i}}(t)]}^ + })}^2} + {{[{Q_{i}}(t)]}^ + }({\varphi {C_{i,f_i}^{\rm th}}(t )} - {u_{i}^{\rm dl}}(t))}\\
\qquad \qquad \quad { + \frac{1}{2}{{({\varphi {C_{i,f_i}^{\rm th}}(t )} - {u_{i}^{\rm dl}}(t))}^2}}
\end{array}
\end{equation}

Case III: when ${{Q_{i}(t + 1)} } < 0$, we can obtain
\begin{equation}\label{eq:31}
\begin{array}{*{20}{l}}
{\frac{1}{2}{{({{[{Q_{i}}(t + 1)]}^ + })}^2} = 0}\\
\qquad \qquad { \le \frac{1}{2}({{[{Q_{i}}(t)]}^ + } + {{({\varphi {C_{i,f_i}^{\rm th}}(t )} - {u_{i}^{\rm dl}}(t))}^2}}\\
\qquad \qquad { = \frac{1}{2}{{({{[{Q_{i}}(t)]}^ + })}^2} + {{[{Q_{i}}(t)]}^ + }({\varphi {C_{i,f_i}^{\rm th}}(t )} - {u_{i}^{\rm dl}}(t))}\\
\qquad \qquad \quad { + \frac{1}{2}{{({\varphi {C_{i,f_i}^{\rm th}}(t )} - {u_{i}^{\rm dl}}(t))}^2}}
\end{array}
\end{equation}

Therefore, we can have
\begin{equation}\label{eq:L_Q}
\begin{array}{*{20}{l}}
{\frac{1}{2}{{({{[{Q_{i}}(t + 1)]}^ + })}^2} \le \frac{1}{2}{{({{[{Q_{i}}(t)]}^ + })}^2}}\\
\qquad \qquad  + {{[{Q_{i}}(t)]}^ + }({\varphi {C_{i,f_i}^{\rm th}}(t )} - {u_{i}^{\rm dl}}(t)) + \\
\qquad \qquad \frac{1}{2}{{({\varphi {C_{i,f_i}^{\rm th}}(t )} - {u_{i}^{\rm dl}}(t))}^2}
\end{array}
\end{equation}

Similarly, according to {(\ref{eq:Queue_Z}), (\ref{eq:Queue_H})}, and the non-negative operation, we have
\begin{equation}\label{eq:L_Z}
\begin{array}{*{20}{l}}
\begin{array}{l}
\frac{1}{2}{({[{Z_{i}}(t + 1)]^ + })^2} = \frac{1}{2}{({[{Z_{i}}(t)]^ + })^2} +
\end{array}\\
\qquad \quad {{[{Z_{i}}(t)]^ + }({\gamma _{i}}(t) - {u_{i}^{\rm dl}}(t)) + \frac{1}{2}{{({\gamma _{i}}(t) - {u_{i}^{\rm dl}}(t))}^2}}
\end{array}
\end{equation}
and
\begin{equation}\label{eq:L_H}
\begin{array}{*{20}{l}}
{\frac{1}{2}{{({{[{H_j}(t + 1)]}^ + })}^2} \le \frac{1}{2}{{({{[{H_j}(t)]}^ + })}^2}}\\
{ + {{[{H_j}(t)]}^ + }({p_j}(t) - {{\tilde p}_j} + p_j^c) + \frac{1}{2}{{({p_j}(t) - {{\tilde p}_j} + p_j^c)}^2}}
\end{array}
\end{equation}

With inequalities (\ref{eq:L_Q})-(\ref{eq:L_H}), we can obtain a new inequality by utilizing the definition of Lyapunov drift. Next, we can achieve (\ref{eq:upper_bound}) by adding $-V\left( g(t)-\sum\nolimits_{j=1}^{N}{{{p}_{j}}(t)} \right)$ to both sides of the new inequality. This completes the proof.

\subsection{Proof of Proposition \ref{lemma:lemma_uav_location}}
For any given UAV transmit power ${\mathcal P}(t)$, UAV trajectories ${{\mathcal {X}}(t-1)}$ at time slot $t-1$, and the content delivery decision set ${\mathcal S}(t)$, we can optimize the variables ${{\mathcal {X}}(t)}$ in (\ref{eq:subproblem_BPX}) via mitigating the following problem
\begin{subequations}\label{eq:UAV_location_problem}
\begin{alignat}{2}
& \mathop {\rm Maximize }\limits_{{{\mathcal {X}}(t)}} {\mkern 1mu} \text{ } {\sum\nolimits_{i} {\{ {{[{Q_{i}}(t)]}^ + } + {{[{Z_{i}}(t)]}^ + }\} {u_{i}^{\rm dl}}(t)} }  \\
&{\rm s.t.: \quad } {(\ref{eq:waypoint_constr}),(\ref{eq:safety_distance}),(\ref{eq:horizontal_dist_constraint}).}
\end{alignat}
\end{subequations}

To simplify (\ref{eq:UAV_location_problem}a), we introduce slack variables $\{ \eta_{i}\}$, with which (\ref{eq:UAV_location_problem}) can be reformulated as
\begin{subequations}\label{eq:UAV_location_equal_problem}
\begin{alignat}{2}
& \mathop {\rm Maximize }\limits_{{{{\mathcal {X}}(t)}, {\{ \eta_{i}}(t) \} }} {\mkern 1mu} \text{ } {\sum\nolimits_{i} {\{ {{[{Q_{i}}(t)]}^ + } + {{[{Z_{i}}(t)]}^ + }\} {\eta _{i}}(t)} }  \\
&{\rm s.t.:} \quad {u_{i}^{\rm dl}}(t) \ge {\eta _{i}}(t),\forall i,t \\
& \quad {(\ref{eq:waypoint_constr}),(\ref{eq:safety_distance}),(\ref{eq:horizontal_dist_constraint}).}
\end{alignat}
\end{subequations}

Denote by $\eta_{i}^{\star}$ the optimal solution to (\ref{eq:UAV_location_equal_problem}). If $\eta_{i}^{\star}$ satisfies (\ref{eq:UAV_location_equal_problem}b) with strict inequality, we can then decrease $u_{i}^{\rm dl}(t)$ to make (\ref{eq:UAV_location_equal_problem}b) active without changing the value of (\ref{eq:UAV_location_equal_problem}a). Therefore, (\ref{eq:UAV_location_equal_problem}) is equivalent to (\ref{eq:UAV_location_problem}).

As $h_{ij}(t)$ can be rewritten as ${h_{ij}}(t) = \frac{{{\theta _{ij}}}}{{{g^2} + ||{{\bm x}_j}(t) - {{\bm x}_{i}^{\rm u}(t)}|{|^2}}}$, where ${\theta _{ij}} = \frac{{G_{LoS}{\varsigma ^2}}}{{16{\pi ^2}}}$, the achievable data rate of user $i$ can be expressed as ${u_{i}^{\rm dl}}(t) = \sum\nolimits_{j } {{s_{ij}}(t){R_{ij}}(t)}$ with
${R_{ij}}(t) = {{\hat R}_{ij}}(t) - {\log _2}( {{\sigma ^2}W + {\sum\nolimits_{k \in {\mathcal J}\backslash \{ j\} }}\frac{{{{p_k}(t)\theta _{ij}}}}{{g^2 + ||{{\bm x}_k}(t) - {{\bm x} _{i}}|{|^2}}}} )$,
and ${\hat R_{ij}}(t) = {\log _2}( {{\sigma ^2}W + \sum\nolimits_{k \in {\mathcal J}} {\frac{{{{p_k}(t)\theta _{ij}}}}{{{g^2} + ||{{\bm x}_k}(t) - {{\bm x}_{i}^{\rm u}(t)}|{|^2}}}} } )$.

(\ref{eq:UAV_location_equal_problem}) is not convex due to the non-convex constraint (\ref{eq:safety_distance}), and (\ref{eq:UAV_location_equal_problem}b). Therefore, we may not find efficient methods to obtain the optimal solution to (\ref{eq:UAV_location_equal_problem}).
Although (\ref{eq:UAV_location_equal_problem}b) is not concave with w.r.t. ${\bm x}_j(t)$, we can observe that $\hat R_{ij}(t)$ is convex w.r.t. ${||{{\bm x}_k}(t) - {{\bm x}_{i}^{\rm u}(t)}||^2}$. Accordingly, a slack variable ${{B}_{ik}(t)}=||{{{\bm x}}_{k}}(t)-{{\bm x}_{i}^{\rm u}(t)}|{{|}^{2}}$ $\forall i, k \ne j$ is involved to transform (\ref{eq:UAV_location_equal_problem}) into the following new problem
\begin{subequations}\label{eq:UAV_location_equal_problem_S}
\begin{alignat}{2}
& \mathop {\rm Maximize }\limits_{{\mathcal {X}}(t), \{ \eta_{i}(t) \}, \{ B_{ik}(t) \}} {\mkern 1mu} \text{ } {\sum\nolimits_{i} {\{ {{[{Q_{i}}(t)]}^ + } + {{[{Z_{i}}(t)]}^ + }\} {\eta _{i}}(t)} }  \\
& {\rm s.t.:} \text{ } \sum\nolimits_{j } {s_{ij}}(t)( {{\hat R}_{ij}}(t) + {{{\tilde R}_{ij}}(t)})  \ge {\eta _{i}}(t), \forall i,t \\
& \quad {B_{ik}(t)} \le ||{{\bm x}_k}(t) - {{\bm x}_{i}^{\rm u}(t)}||{^2},\forall i,k \ne j,t \\
& \quad {(\ref{eq:waypoint_constr}),(\ref{eq:safety_distance}),(\ref{eq:horizontal_dist_constraint}).}
\end{alignat}
\end{subequations}
where ${\tilde R_{ij}}(t) =  - {\log _2}( {{\sigma ^2}W + \sum\nolimits_{k \in {\mathcal J}\backslash \{ j\} } {\frac{{{{p_k}(t)\theta _{ij}}}}{{{g^2} + {B_{ik}}(t)}}} } )$.

Similar to (\ref{eq:UAV_location_equal_problem}), although a slack variable $B_{ik}(t)$ is introduced, (\ref{eq:UAV_location_equal_problem_S}) is equivalent to (\ref{eq:UAV_location_equal_problem}). Unfortunately, (\ref{eq:UAV_location_equal_problem_S}) is still non-convex as (\ref{eq:safety_distance}), (\ref{eq:UAV_location_equal_problem_S}b), and (\ref{eq:UAV_location_equal_problem_S}c) are non-convex.

To handle the non-convexity of (\ref{eq:UAV_location_equal_problem_S}), an SCA technique is explored. It can be observed that $\hat R_{ij}(t)$, $\forall i,j$ is convex w.r.t. ${||{{\bm x}_k}(t) - {{\bm x}_{i}^{\rm u}(t)}||^2}$ and will be globally lower-bounded by its first-order Taylor expansion at any local point \cite{Stephen2004Convex}. Therefore, for a given local point at the $(r+1)$-th iteration ($r \ge 0$), denoted by ${{\bm x}_k^{(r)}(t)}$, $\hat R_{ij}(t)$ is lower-bounded by
\begin{equation}\label{eq:Rij}
\begin{array}{l}
{{\hat R}_{ij}}(t) \ge {\log _2}\left( {{\sigma ^2}W + \sum\nolimits_{k \in {\mathcal J}} {\frac{{{{p_k}(t)\theta _{ij}}}}{{{g^2} + ||{\bm x}_k^{(r)}(t) - {{\bm x}_{i}^{\rm u}(t)}|{|^2}}}} } \right)\\
 - \sum\limits_{k \in {\mathcal J}} {\frac{{\frac{{{{p_k}(t)\theta _{ij}}}}{{{{\left( {{g^2} + ||{\bm x}_k^{(r)}(t) - {{\bm x}_{i}^{\rm u}(t)}|{|^2}} \right)}^2}}}\left( {||{{\bm x}_k}(t) - {{\bm x}_{i}^{\rm u}(t)}|{|^2} - ||{\bm x}_k^{(r)}(t) - {{\bm x}_{i}^{\rm u}(t)}|{|^2}} \right)}}{{\left( {{\sigma ^2}W + \sum\limits_{k \in {\mathcal J}} {\frac{{{{p_k}(t)\theta _{ij}}}}{{{g^2} + ||{\bm x}_k^{(r)}(t) - {{\bm x}_{i}^{\rm u}(t)}|{|^2}}}} } \right)\ln 2}}} \\
 = D_{i}^{(r)}(t) - \sum\limits_{k \in {\mathcal J}} {E_{ik}^{(r)}(t)( {||{{\bm x}_k}(t) - {{\bm x}_{i}^{\rm u}(t)}|{|^2} -}} \\
 \qquad \qquad \qquad {{||{\bm x}_k^{(r)}(t) - {{\bm x}_{i}^{\rm u}(t)}|{|^2}} )}
\end{array}
\end{equation}
where $D_{i}^{(r)}(t) = {\log _2}\left( {{\sigma ^2}W + \sum\limits_{k \in {\mathcal J}} {\frac{{{{p_k}(t)\theta _{ij}}}}{{{g^2} + ||{\bm x}_k^{(r)}(t) - {{\bm x}_{i}^{\rm u}(t)}|{|^2}}}} } \right)$ and $E_{ik}^{(r)}(t) = {{\frac{{{{p_k}(t)\theta _{ij}}}}{{{{\left( {{g^2} + ||{\bm x}_k^{(r)}(t) - {{\bm x}_{i}^{\rm u}(t)}|{|^2}} \right)}^2{{2^{D_{i}^{(r)}(t)}\ln 2}}}}}}}$.

Besides, for a given location point $({\bm x}_{j}^{(r)}(t), {\bm x}_{k}^{(r)}(t))$, we can obtain the lower bound of ${{\left\| {{{\bm x}}_{j}}(t)-{{{\bm x}}_{k}}(t) \right\|}^{2}}$ via the first order Taylor expansion as below
\begin{equation}\label{eq:xjk}
\begin{array}{*{20}{l}}
{|| {{{\bm x}_j}(t) - {{\bm x}_k}(t)} ||^2} \ge  - {|| {{\bm x}_j^{(r)}(t) - {\bm x}_k^{(r)}(t)} ||^2} + \\
\qquad \qquad \quad 2{( {{\bm x}_j^{(r)}(t) - {\bm x}_k^{(r)}(t)} )^{\rm T}}\left( {{{\bm x}_j}(t) - {{\bm x}_k}(t)} \right).
\end{array}
\end{equation}

Similarly, for a given location point ${\bm x}_k^{r}(t)$, $||{{\bm x}_k}(t) - {{\bm x}_{i}^{\rm u}(t)}||{^2}$ is lower-bounded by
\begin{align}\label{eq:xk}
& ||{{\bm x}_k}(t) - {{\bm x}_{i}^{\rm u}(t)}||{^2} \ge  {|| {{\bm x}_k^{(r)}(t) - {\bm x}_{i}^{\rm u}(t)} ||^2} + \notag \\
& \qquad \qquad 2{( {{\bm x}_k^{(r)}(t) - {\bm x}_{i}^{\rm u}(t)} )^{\rm T}} \left( {{{\bm x}_k}(t) - {{\bm x}_{i}^{\rm u}(t)}} \right).
\end{align}

For any local point ${\mathcal X}^{(r)}(t) = \{{\bm x}_k^{(r)}(t)\}$, by referring to (\ref{eq:Rij})-(\ref{eq:xk}), (\ref{eq:UAV_location_equal_problem_S}) can be approximated as (\ref{eq:UAV_location_equal_problem_approximate}). This completes the proof.

\subsection{Proof of Proposition \ref{lemma:lemma_UAV_power}}
For any given content delivery decision set ${\mathcal S}(t)$ as well as UAV trajectories ${\mathcal X}(t)$, the UAV transmit power of (\ref{eq:original_problem}) can be optimized via mitigating the following problem
\begin{subequations}\label{eq:UAV_power_problem}
\begin{alignat}{2}
& \mathop {\rm Maximize }\limits_{{\mathcal {P}}(t),\{ \eta_{i}(t) \}} {\mkern 1mu} \text{ } - V\rho \sum\nolimits_{j } {{p_j}(t)}  - \sum\nolimits_{j } {{{[{H_j}(t)]}^ + }{p_j}(t)}  + \nonumber \\
& \qquad \qquad \sum\nolimits_{i} {\{ {{[{Q_{i}}(t)]}^ + } + {{[{Z_{i}}(t)]}^ + }\} {\eta _{i}}(t)} \\
& {\rm s.t.:} \text{ } \sum\limits_{j } {{s_{ij}}(t){{\log }_2}( {1 + \frac{{{p_j}(t){h_{ij}}(t)}}{{{\sigma ^2}W + \sum\nolimits_{k \in {\mathcal J}\backslash \{ j\} } {{p_k}(t){h_{ik}}(t)} }}} )}   \nonumber \\
& \qquad \qquad \quad \ge {\eta _{i}}(t), \quad i,t \\
& \quad \eta_i(t) \ge { s_{ij}(t){C_{i,f_i}^{\rm th}}(t )}, \quad i,j,t \\
& \quad {\rm (\ref{eq:original_problem}d), (\ref{eq:original_problem}g)}
\end{alignat}
\end{subequations}
where (\ref{eq:UAV_power_problem}c) is enforced due to the data rate requirement of enabling a user's desired QoE state.

Owing to the non-convex constraint (\ref{eq:UAV_power_problem}b), (\ref{eq:UAV_power_problem}) is non-convex; as a result, it is challenging to achieve its optimal solution. However, we observe that (\ref{eq:UAV_power_problem}b) is a difference of two concave functions w.r.t. $p_k(t)$. Accordingly, we adopt the SCA technique again to approximate (\ref{eq:UAV_power_problem}b). Specifically, $R_{ij}(t)$ can be rewritten as ${R_{ij}}(t) = {\hat R_{ij}}(t) - {{\overset{\scriptscriptstyle\smile}{R}}_{ij}}(t)$, where ${{\overset{\scriptscriptstyle\smile}{R}}_{ij}}(t) = {\log _2}\left( {{\sigma ^2}W + \sum\nolimits_{k \in {\mathcal J}\backslash \{ j\} } {{p_k}(t){h_{ik}}(t)} } \right)$. For any local point ${\mathcal P}^{(r)}(t) = \{p_j^{r}(t)\}$, via the first order Taylor expansion ${{\overset{\scriptscriptstyle\smile}{R}}_{ij}}(t)$ is upper-bounded by
\begin{equation}\label{eq:8}
\begin{array}{l}
{{\mathord{\buildrel{\lower3pt\hbox{$\scriptscriptstyle\smile$}}
\over R} }_{ij}}(t) \le {\log _2}\left( {{\sigma ^2}W + \sum\nolimits_{k \in {\mathcal J}\backslash \{ j\} } {p_k^{(r)}(t){h_{ik}}(t)} } \right)\\
 + \sum\limits_{k \in {\mathcal J}\backslash \{ j\} } {\frac{{{h_{ik}}(t)\ln^{-1} 2}}{{ {{\sigma ^2}W + \sum\nolimits_{k \in {\mathcal J}\backslash \{ j\} } {p_k^{(r)}(t){h_{ik}}(t)} } }}} \left( {{p_k}(t) - p_k^{(r)}(t)} \right)\\
 = F_{ij}^{(r)}(t) + \sum\nolimits_{k \in {\mathcal J}\backslash \{ j\} } {G_{ik}^{(r)}(t)} \left( {{p_k}(t) - p_k^{(r)}(t)} \right)
\end{array}
\end{equation}
where $F_{ij}^{(r)}(t) = {\log _2}\left( {{\sigma ^2}W + \sum\nolimits_{k \in {\mathcal J}\backslash \{ j\} } {p_k^{(r)}(t){h_{ik}}(t)} } \right)$ and $G_{ik}^{(r)}(t) = \frac{{{h_{ik}}(t)}}{{2^{F_{ij}^{(r)}(t)}\ln 2}}$.


We can thus write the lower bound of $R_{ij}(t)$ as ${R_{ij}}(t) \ge {\hat R}_{ij}(t) - F_{ij}^{(r)}(t) - \sum\nolimits_{k \in {\mathcal J}\backslash \{ j\} } {G_{ik}^{(r)}(t)} ( {{p_k}(t) - p_k^{(r)}(t)} )
$.

In summary, for any local point ${\mathcal P}^{(r)}(t)$, (\ref{eq:UAV_power_problem}) can be approximated as (\ref{eq:UAV_power_problem_approx}). This completes the proof.

\subsection{Proof of Lemma \ref{lemma:lemma_convergent}}
Given a local point $({\mathcal X}^{(r)}(t), {\mathcal P}^{(r)}(t))$, the obtained value of (\ref{eq:UE_association_problem}a) at the $(r+1)$-th iteration, denoted by $\Gamma({{\mathcal S}^{(r+1)}(t)},{{\mathcal X}^{(r)}(t)},{{\mathcal P}^{(r)}(t)})$, is no greater than $\Gamma({{\mathcal S}^{(r)}(t)},{{\mathcal X}^{(r)}(t)},{{\mathcal P}^{(r)}(t)})$ via optimizing (\ref{eq:UE_association_problem}). Given a point $({\mathcal S}^{(r+1)}(t), {\mathcal P}^{(r)}(t))$, we have $\Gamma({{\mathcal S}^{(r+1)}(t)},{{\mathcal X}^{(r)}(t)},{{\mathcal P}^{(r)}(t)}) \ge \Gamma({{\mathcal S}^{(r+1)}(t)},{{\mathcal X}^{(r+1)}(t)},{{\mathcal P}^{(r)}(t)})$ due to the minimization of the upper-bounded problem of (\ref{eq:UAV_location_equal_problem}). Likewise, the inequality $\Gamma({{\mathcal S}^{(r+1)}(t)},{{\mathcal X}^{(r+1)}(t)},{{\mathcal P}^{(r+1)}(t)}) \ge \Gamma({{\mathcal S}^{(r+1)}(t)},{{\mathcal X}^{(r+1)}(t)},{{\mathcal P}^{(r)}(t)})$ can be obtained at $({\mathcal S}^{(r+1)}(t), {\mathcal X}^{(r+1)}(t))$. Besides, $\Gamma({{\mathcal S}^{(r)}(t)},{{\mathcal X}^{(r)}(t)},{{\mathcal P}^{(r)}(t)})$ is bounded at each iteration. Therefore, Algorithm \ref{alg:alg1} is convergent.

Lemma \ref{lemma:1} points out that $\Delta (t)-V( g(t)-\rho \sum\nolimits_{j }{{{p}_{j}^{tot}}(t)} )$ is upper-bounded at each time slot $t$. The time average of $L(t)$ then tends to be zero when $t \to \infty$. Therefore, Algorithm \ref{alg:alg2} can make all virtual queues mean-rate stable. This completes the proof.

\subsection{Proof of Lemma \ref{lem:lemma_accumulated}}
According to the evolution model in (\ref{eq:queue_evolution}), ${N_{a}^1}=0$.
Therefore, we focus on deriving the close-form expression of ${N_{a}^q}$ with $q>1$.
During the $2$-nd time interval, $n = 0$ indicates that there are no new arrivals or $n_c$ new arrivals during the $1$-st time interval; $n > 0$ means that the number of new arrivals in the $1$-st time interval is $n+n_c$. Therefore, the probability mass function (PMF) of the accumulated packets, denoted as ${f_{N_a^2}}(n)$, can be expressed as
\begin{equation}\label{eq:accumu_2nd}
{f_{N_a^2}}(n) = \left\{ {\begin{array}{*{20}{l}}
{{e^{ - \vartheta _w^1}} + \vartheta _w^1{e^{ - \vartheta _w^1}},}&{n = 0}\\
{\frac{{{{\left( {\vartheta _w^1} \right)}^{n + n_c}}{e^{ - \vartheta _w^1}}}}{{\left( {n + n_c} \right)!}},}&{n > 0}
\end{array}} \right.
\end{equation}

Similarly, in the $3$-rd time interval, three cases will result in $n = 0$: Case I, both the accumulate number of packets and the number of new arrivals in the $2$-nd time interval are zero; Case II, the accumulated number of packets is $n_c$ and the number of new arrivals is zero in the $2$-nd time interval; Case III, the accumulated number of packets is zero while the number of new arrivals is $n_c$ in the $2$-nd time interval. $n > 0$ indicates that the sum of the accumulated number of packets and the number of new arrivals in the $2$-nd time interval is $n+n_c$. Therefore, we can obtain the PMF of the accumulate number of packets, denoted by ${f_{N_a^3}}(n)$, in the $3$-rd time interval as follows
\begin{equation}\label{eq:accumu_$3$-rd}
{f_{N_a^3}}(n) = \left\{ {\begin{array}{*{20}{l}}
\begin{array}{l}
{e^{ - \vartheta _w^2}}{f_{N_a^2}}(0) + \vartheta _w^2{e^{ - \vartheta _w^2}}{f_{N_a^2}}(0) + \\
{e^{ - \vartheta _w^2}}{f_{N_a^2}}(n_c),
\end{array}&{n = 0}\\
{\sum\limits_{z = 0}^{n + n_c} {\frac{{{{\left( {\vartheta _w^2} \right)}^z}}}{{z!}}{e^{ - \vartheta _w^2}}{f_{N_a^2}}(n + n_c - z)} ,}&{n > 0}
\end{array}} \right.
\end{equation}

Likewise, in the $q$-th time interval, we can obtain the PMF of the accumulated number of packets, denoted by ${f_{N_a^q}}(n)$, as follows
\begin{equation}\label{eq:accumu_tth}
{f_{N_a^q}}(n) = \left\{ {\begin{array}{*{20}{l}}
\begin{array}{l}
{e^{ - \vartheta _w^{q-1}}}{f_{N_a^{q-1}}}(0) + \vartheta _w^{q-1}{e^{ - \vartheta _w^{q-1}}} \times \\
{f_{N_a^{q-1}}}(0) + {e^{ - \vartheta _w^{t-1}}}{f_{N_a^{q-1}}}(n_c),
\end{array}&{n = 0}\\
{\sum\limits_{z = 0}^{n + n_c} {\frac{{{{\left( {\vartheta _w^{q-1}} \right)}^z}}}{{z!}}{e^{ - \vartheta _w^{q-1}}}{f_{N_a^{q-1}}}(n + n_c - z)} ,}&{n > 0}
\end{array}} \right.
\end{equation}

In (\ref{eq:accumu_tth}), ${f_{N_a^q}}(n)$ correlates with ${f_{N_a^{q-1}}}(n+n_c-z)$ in a sophisticated recursive way, which significantly hinders the theoretical derivation of the closed-form expression of ${f_{N_a^q}}(n)$. Moreover, the complexity of the theoretical derivation exponentially increases with $q$. To tackle this problem, we propose to derive an approximated expression of ${f_{N_a^q}}(n)$.
As new packets arrive in a Poisson process, the packet departure can be considered as an approximated packet thinning of the arrived packets \cite{hohn2006inverting,jiang2018random}. After this packet thinning in a specific time interval, the number of accumulated packets in time $q$ $(q > 1)$ can be approximated as a Poisson distribution \cite{jiang2018random}. Then, denote by $\vartheta_{a}^q$ the average number of accumulated packets in time $q$. In the $2$-nd time interval, we can calculate $\vartheta_{a}^2$ as
\begin{equation}\label{eq:intensity_$2$-nd}
\begin{array}{l}
\vartheta _a^2\mathop  = \limits^{(a)} \sum\limits_{n = 1}^\infty  {(n - n_c)\frac{{{{\left( {\vartheta _w^1} \right)}^n}{e^{ - \vartheta _w^1}}}}{{n!}}} \\
 = \sum\limits_{n = 0}^\infty  {n\frac{{{{\left( {\vartheta _w^1} \right)}^n}{e^{ - \vartheta _w^1}}}}{{n!}}}  - n_c\left( {\sum\limits_{n = 0}^\infty  {\frac{{{{\left( {\vartheta _w^1} \right)}^n}{e^{ - \vartheta _w^1}}}}{{n!}}}  - {e^{ - \vartheta _w^1}}} \right)\\
\mathop  = \limits^{(b)} {\left[ {\vartheta _w^1 - n_c\left( {1 - {e^{ - \vartheta _w^1}}} \right)} \right]^ + }
\end{array}
\end{equation}
where (a) holds because ${\rm Pr}\{N_a^2=n-n_c\} = {\rm Pr}\{N_w^1=n\}$ if $n \ne 0$ with ${\rm Pr}\{x\}$ denoting the probability of event $x$, and (b) holds because $\vartheta _a^2$ is non-negative at each time slot $t$.

In the $3$-rd time interval, the intensity of accumulated data packets in the preprocessing queue can be derived as the following
\begin{equation}\label{eq:mu_acc_3}
\begin{array}{l}
\vartheta _a^3 = \sum\limits_{n = 1}^\infty  {(n - n_c)\sum\limits_{z = 0}^n {\frac{{{{\left( {\vartheta _w^2} \right)}^z}{e^{ - \vartheta _w^2}}}}{{z!}}\frac{{{{\left( {\vartheta _a^2} \right)}^{n - z}}{e^{ - \vartheta _a^2}}}}{{(n - z)!}}} } \\
 = \sum\limits_{n = 1}^\infty  {n\sum\limits_{z = 0}^n {\frac{{{{\left( {\vartheta _w^2} \right)}^z}{e^{ - \vartheta _w^2}}}}{{z!}}\frac{{{{\left( {\vartheta _a^2} \right)}^{n - z}}{e^{ - \vartheta _a^2}}}}{{(n - z)!}}} }  - \\
n_c\left( {\sum\limits_{n = 0}^\infty  {\sum\limits_{z = 0}^n {\frac{{{{\left( {\vartheta _w^2} \right)}^z}{e^{ - \vartheta _w^2}}}}{{z!}}\frac{{{{\left( {\vartheta _a^2} \right)}^{n - z}}{e^{ - \vartheta _a^2}}}}{{(n - z)!}}} }  - {e^{ - \vartheta _w^2 - \vartheta _a^2}}} \right)\\
 = {\left[ {\vartheta _w^2 + \vartheta _a^2 - n_c\left( {1 - {e^{ - \vartheta _w^2 - \vartheta _a^2}}} \right)} \right]^ + }
\end{array}
\end{equation}

When $q > 3$, since the accumulated packet evolution model of the preprocessing queue is similar to that at $q = 3$,
we can extend the conclusion obtained at $q = 3$ to that at $q > 3$. Therefore, we can obtain the closed-form expression of $\vartheta_a^q$ at $q > 3$ with
\begin{equation}\label{eq:mu_acc_t}
\vartheta _a^q = {\left[ {\vartheta _w^{q-1} + \vartheta _a^{q-1} - n_c\left( {1 - {e^{ - \vartheta _w^{q-1} - \vartheta _a^{q-1}}}} \right)} \right]^ + }
\end{equation}

Besides, at time $q$, for packet $m$, it has to wait until the completion of the preprocessing of accumulated packets in the preprocessing queue according to the FCFS principle. The preprocessing time of each packet is $Y_{i,m}^{\rm S}$ due to the same packet size and preprocessing operation. Thus, the average packet queueing delay of packet $m$ is $\vartheta_{a}^q Y_{i,m}^{\rm S}$. This completes the proof.

\subsection{Proof of Proposition \ref{prop:prop_low_latency}}
For packet $m$, it will be delivered to its receiving user $i$ if and only if a) $m$ has been processed by the BS; b) $\sum\nolimits_{j } s_{ij}(t) = 1$ at time slot $t$.
As there are $N$ users and $J$ UAVs and the goal of the communication problem is to provide fair content delivery for all users, the probability that any user can connect to a UAV is $J/N$ at each time slot.
The delivery of a content file including $L/l$ packets from the BS or a UAV to its corresponding user should be completed in a time slot of duration $\Delta_ t$.
Further, one packet will be sent to user $i$ at a time, and $L/l$ packets will be sequentially delivered in duration $\Delta_ t$. Then, we have ${\mathbb E}[Y_{i,1}^{\rm A}] = \frac{Nl\Delta_ t}{JL}$, and the expected edge arrival duration of the last packet of the content file is $\frac{N\Delta_ t}{J}$.
Thus, for any packet $m$ in the content file transmitted to user $i$, the expected edge arrival duration of packet $m$ is $\frac{Nl\Delta_ t}{2JL}+\frac{N\Delta_ t}{2J}$. This completes the proof.

\subsection{Proof of Lemma \ref{lemma:EPAoI}}
By referring to (\ref{eq:AoI_model}), we can write the expectation of PAoI of packet $m$ sent to $i$ as follows
\begin{equation}\label{eq:average_aoi}
\begin{array}{l}
{\Delta _{i,m}}(\hat q; q) = {\mathbb E}[{X_{i,m}}] + {\mathbb E}[{Y_{i,m}}]\\
 = {\mathbb E}[{X_{i,m}}] + {\mathbb E}[Y_{i,m}^{\rm Q}] + {\mathbb E}[Y_{i,m}^{\rm S}] + {\mathbb E}[Y_{i,m}^{\rm A}]
\end{array}
\end{equation}

Next, we discuss the value of ${\Delta _{i,m}}(\hat q;q)$ from the following two cases:

Case I: $q = 1$ and $m = 1$. The value of ${\rm Pr}\{M(q)=0\}$ will become smaller as $q$ increases, where $M(q)$ denotes the number of newly arrived packets before time $q$. For example, according to (\ref{eq:intensity_$2$-nd}), the average number of accumulated packets in the preprocessing queue is ${[ {\vartheta _w^1 - n_c( {1 - {e^{ - \vartheta _w^1}}} )} ]^ + }$ when $q = 2$.
Besides, the content server will generate packets required by a user after receiving its content request. It indicates that the probability that the content server generates packets destined to $i$ during each time interval is equal to ${\rm Pr}_i$, $\forall i$.
Then, the expected number of new packets is ${ {\vartheta _w^1 + n_c { {e^{ - \vartheta _w^1}}} }  }$, and the number of packets sent to user $i$ is $M_{i}(q)={( {\vartheta _w^1 + n_c { {e^{ - \vartheta _w^1}}} } ) }{\rm Pr}_i$, where $M_i(q)$ is the expected number of packets destined to user $i$. 
We can envision that the expected number of generated packets sent to user $i$ will be greater than $({ { \vartheta _w^1 + n_c { {e^{ - \vartheta _w^1 }}} }  }){\rm Pr}_i$ when $q>2$.
Recall that the content server generates packets according to a Poisson process with rate $\vartheta_w^q$ at time $q$ and sends them to the BS.
The generated packets can be considered to be belonging to $N$ different data streams for $N$ terrestrial users, respectively. Each data stream $i$ is chosen independently at time $q$ with probability ${\rm Pr}_i$.
This setup is equivalent to having $N$ independent Poisson sources with rates $\lambda_i^q = \vartheta_w^q {\rm Pr}_i$, $\forall i$, and $\vartheta_w^q = \lambda_1^q + \cdots + \lambda_N^q$ (see \cite{ross1996stochastic}).
Thus, we have ${\rm Pr}\{N_i(q)=m\}=\frac{(\vartheta_w^q q{\rm Pr}_i)^m}{m!}e^{-\vartheta_w^q q{\rm Pr}_i}$ and ${\rm Pr}\{N_i(1)=0\}=1/e^{\vartheta_w^1 {\rm Pr}_i}$, which will be a small value. It indicates that the event that the content server does not generate the first packet for user $i$ in the $1$-st time interval is a small probability event. We therefore consider the case of $q=1$ and $m = 1$.
On the other hand, according to (\ref{eq:queue_evolution}), the accumulated packets in the preprocessing queue is zero when $q=1$. Therefore, the expected time for the first packet destined to user $i$ includes the preprocessing time and the expected edge arrival duration, i.e., we can obtain
\begin{equation}\label{eq:m_1_and_q_1}
\begin{array}{l}
{\mathbb E}[{\Delta _{i,1}}({{\hat 1}; 1})] = {\mathbb E}[Y_{i,1}^{\rm S}] + {\mathbb E}[Y_{i,m}^{\rm A}] \\
= \frac{ 1}{n_c} + \frac{(l+L)N\Delta_ t}{2JL}
\end{array}
\end{equation}

Case II: $q >1$ and $m > 1$.
As the inter-arrival time of packets destined to user $i$ from the content server follows an exponential distribution with intensity $\lambda_i^q$, we have ${\mathbb E}[X_{i,m}] = 1/\lambda_i^q$.
Besides, since each packet has the same size, the expected preprocessing time for each packet will be the same, i.e., ${\mathbb E}[Y_{i,m}^{\rm S}] = 1/n_c$.
From Lemma \ref{lem:lemma_accumulated}, we know that the average packet queueing delay of packet $m$ is $\vartheta_{a}^q Y_{i,m}^{\rm S}$. Thus, we have ${\mathbb E}[Y_{i,m}^{\rm Q}] = \vartheta _a^q /n_c$.
Proposition \ref{prop:prop_low_latency} also shows that ${\mathbb E}[Y_{i,m}^{\rm A}] = \frac{(l+L)N\Delta_ t}{2JL}$ $\forall i, m$.
Therefore, (\ref{eq:average_PAoI}) is obtained. This completes the proof.

\ifCLASSOPTIONcaptionsoff
  \newpage
\fi




%
\bibliographystyle{IEEEtran}
\bibliography{UAV_Cache}

\end{document}